\documentclass[11pt]{article}

\usepackage[letterpaper,margin=0.9in]{geometry}
\usepackage{amsmath, amssymb, amsthm, thmtools, amsfonts}
\usepackage{bbm}

\usepackage{ifthen}

\usepackage{tikz}
\usetikzlibrary{positioning,decorations.pathreplacing}

\usepackage{cite}
\usepackage{appendix}
\usepackage{graphicx}
\usepackage{color}
\usepackage[boxed]{algorithm2e}
\usepackage[noend]{algpseudocode}
\usepackage{epstopdf}
\usepackage[textsize=tiny]{todonotes}
\usepackage{framed}
\usepackage[framemethod=tikz]{mdframed}
\usepackage[bottom]{footmisc}
\usepackage[shortlabels]{enumitem}
\setitemize{noitemsep,topsep=3pt,parsep=3pt,partopsep=3pt}
\usepackage[font=small]{caption}
\usepackage{xspace}

\newtheorem{theorem}{Theorem}[section]
\newtheorem{lemma}[theorem]{Lemma}
\newtheorem{meta-theorem}[theorem]{Meta-Theorem}

\newtheorem{remark}[theorem]{Remark}
\newtheorem{corollary}[theorem]{Corollary}

\newtheorem{definition}[theorem]{Definition}

\definecolor{darkgreen}{rgb}{0,0.5,0}
\usepackage{hyperref}
\hypersetup{
    unicode=false,          % non-Latin characters in Acrobat’s bookmarks
    colorlinks=true,        % false: boxed links; true: colored links
    linkcolor=red,          % color of internal links (change box color with linkbordercolor)
    citecolor=darkgreen,        % color of links to bibliography
    filecolor=magenta,      % color of file links
    urlcolor=cyan           % color of external links
}
\usepackage[capitalize, nameinlink]{cleveref}
\crefname{theorem}{Theorem}{Theorems}
\Crefname{lemma}{Lemma}{Lemmas}

\algnewcommand\algorithmicswitch{\textbf{switch}}
\algnewcommand\algorithmiccase{\textbf{case}}

\algdef{SE}[SWITCH]{Switch}{EndSwitch}[1]{\algorithmicswitch\ #1\ \algorithmicdo}{\algorithmicend\ \algorithmicswitch}%
\algdef{SE}[CASE]{Case}{EndCase}[1]{\algorithmiccase\ #1}{\algorithmicend\ \algorithmiccase}%
\algtext*{EndSwitch}%
\algtext*{EndCase}%
\newcommand{\eps}{\varepsilon}

\newcommand{\congest}{$\mathsf{CONGEST}$\xspace}
\newcommand{\local}{$\mathsf{LOCAL}$\xspace}

\newcommand{\poly}{\operatorname{\text{{\rm poly}}}}

\newcommand{\ceil}[1]{\lceil #1 \rceil}

\newcommand{\set}[1]{\left\{#1\right\}}

\newcommand{\round}{\mathsf{round}}

\newcommand{\paren}[1]{\mathopen{}\left(#1\right)\mathclose{}}

\renewcommand{\paragraph}[1]{\vspace{0.15cm}\noindent {\bf #1}:}

\newcommand{\FullOrShort}{full}

\ifthenelse{\equal{\FullOrShort}{full}}{
	
  \newcommand{\fullOnly}[1]{#1}
  \newcommand{\shortOnly}[1]{}
  
  }{

    \newcommand{\fullOnly}[1]{}
    \newcommand{\IncludePictures}[1]{}
   
  }

\begin{document}

\date{}

\title{Deterministic Distributed Edge-Coloring \\ via Hypergraph Maximal Matching
}

\author{
	 Manuela Fischer\\
  \small ETH Zurich \\
  \small manuela.fischer@inf.ethz.ch
\and
 Mohsen Ghaffari\\
  \small ETH Zurich \\
  \small ghaffari@inf.ethz.ch
\and 
Fabian Kuhn\\
  \small University of Freiburg \\
  \small kuhn@cs.uni-freiburg.de
 }

\maketitle

\begin{abstract}
We present a deterministic distributed algorithm that computes a $(2\Delta-1)$-edge-coloring, or even list-edge-coloring, in any $n$-node graph with maximum degree $\Delta$, in $O(\log^7 \Delta \cdot \log n)$ rounds. This answers one of the long-standing open questions of \emph{distributed graph algorithms} from the late 1980s, which asked for a polylogarithmic-time algorithm.  See, e.g., Open Problem 4 in the Distributed Graph Coloring book of Barenboim and Elkin. The previous best round complexities were $2^{O(\sqrt{\log n})}$ by Panconesi and Srinivasan [STOC'92] and $\tilde{O}(\sqrt{\Delta}) + O(\log^* n)$ by Fraigniaud, Heinrich, and Kosowski [FOCS'16]. A corollary of our deterministic list-edge-coloring also improves the randomized complexity of $(2\Delta-1)$-edge-coloring to $\poly(\log\log n)$ rounds.

\medskip
The key technical ingredient is a deterministic distributed algorithm for \emph{hypergraph maximal matching}, which we believe will be of interest beyond this result. In any hypergraph of rank $r$ --- where each hyperedge has at most $r$ vertices --- with $n$ nodes and maximum degree $\Delta$, this algorithm computes a maximal matching in $O(r^5 \log^{6+\log r } \Delta \cdot \log n)$ rounds. 

\medskip
This hypergraph matching algorithm and its extensions also lead to a number of other results. In particular, we obtain a polylogarithmic-time deterministic distributed maximal independent set (MIS) algorithm for graphs with bounded neighborhood independence, hence answering Open Problem 5 of Barenboim and Elkin's book, a $\big((\log \Delta/\eps)^{O(\log 1/\eps)}\big)$-round deterministic algorithm for $(1+\eps)$-approximation of maximum matching, and a quasi-polylogarithmic-time deterministic distributed algorithm for orienting $\lambda$-arboricity graphs with out-degree at most $\ceil{(1+\eps)\lambda}$, for any constant $\eps>0$, hence partially answering Open Problem 10 of Barenboim and Elkin's book.

\end{abstract}

\setcounter{page}{0}
\thispagestyle{empty}
\newpage

\section{Introduction and Related Work}
\vspace{-8pt}
\emph{Distributed graph algorithms} have been studied extensively over the past 30 years, since the seminal work of Linial\cite{linial1987LOCAL}. Despite this, determining whether there are efficient deterministic distributed algorithms for the most classic problems of the area remains a long-standing open question. 

Distributed graph algorithms are typically studied in a standard synchronous message passing model known as the \local\ model \cite{linial1987LOCAL,Peleg:2000}: the network is abstracted as an undirected graph $G=(V, E)$, $n=|V|$, with maximum degree $\Delta$, and where each node has a $\Theta(\log n)$-bit unique identifier. Initially, nodes only know their neighbors in $G$. At the end, each node should know its own part of the solution, e.g., the colors of its edges in edge-coloring. Communication happens in synchronous rounds, where in each round each node sends a message to each of its neighbors.\footnote{In the \local\ model, messages might be of arbitrary size. A variant of the model where the messages have to be of bounded size is known as the \congest\ model \cite{Peleg:2000}. Our edge-coloring algorithms and our MIS and vertex-coloring algorithms for graphs of bounded neighborhood independence in fact work with small $O(\log n)$-bit size messages. Though, for the sake of the readability, we avoid explicitly discussing the details of this aspect.} The main complexity measure is the number of rounds needed for solving a given graph problem.

The four classic local distributed graph problems are maximal independent set (MIS), ($\Delta+1$)-vertex-coloring, ($2\Delta-1$)-edge-coloring, and maximal matching\cite{panconesirizzi2000, barenboim2013monograph}. All of these problems have trivial greedy sequential algorithms, as well as simple $O(\log n)$-round randomized distributed algorithms\cite{luby1986simple, alon1986fast}, and even some faster ones\cite{barenboim2012locality, elkin2015EdgeColoring, Ghaffari-MIS, harris2016distributed}. But the deterministic distributed complexity of these problems remains widely open, despite extensive interest (see, e.g., the first five open problems of \cite{barenboim2013monograph}). Particularly, with regards to MIS --- which is the hardest of the four problems, as the other three can be reduced to MIS locally\cite{linial1987LOCAL} --- Linial\cite{linial1992locality} asked \begin{center} \begin{minipage}{0.85\linewidth}
\vspace{-5pt}
\begin{mdframed}[hidealllines=true, backgroundcolor=gray!00]
\emph{``can it [MIS] always be found [deterministically] in polylogarithmic time?".}  
\vspace{-7pt}
\end{mdframed}
\end{minipage}
\end{center}
This remains the most well-known open question of the area. The best known round complexity is $2^{O(\sqrt{\log n})}$, due to Panconesi and Srinivasan\cite{panconesi1992improved}. Panconesi and Rizzi pointed out in the year 2000 \cite{panconesirizzi2000} that ``\emph{while maximal matchings can be computed in polylogarithmic time, in $n$, in the distributed model\cite{HanckowiakKP98}, it is a decade old open problem whether the same running time is achievable for the remaining 3 structures.}" The status remains the same as of today, after almost two more decades. In particular, for edge-coloring which is our main target, Barenboim and Elkin stated the following problem in their recent Distributed Graph Coloring book\cite{barenboim2013monograph}: 

\begin{center}
\begin{minipage}{0.95\linewidth}
\vspace{-6pt}
\begin{mdframed}[hidealllines=true, backgroundcolor=gray!00]
\textsc{Open Problem 11.4}\cite{barenboim2013monograph} $\;$ Devise or rule out a deterministic ($2\Delta-1$)-edge-coloring algorithm that runs in polylogarithmic time. 
\vspace{-8pt}
\end{mdframed}
\end{minipage}
\end{center}

\subsection{Our Contributions}
\vspace{-8pt}
\paragraph{Improved Deterministic Edge-Coloring Algorithm} One of our main end results is a positive resolution of the above question.

\vspace{-3pt}
\begin{theorem}\label{thm:edge-coloring} There is a deterministic distributed algorithm that computes a $(2\Delta-1)$-edge-coloring in $O(\log^{7} \Delta \cdot \log n)$ rounds, in any $n$-node graph with maximum degree $\Delta$. Moreover, the same algorithm solves list-edge-coloring, where each edge $e\in E$ must get a color from an arbitrary given list $L_e$ of colors with $|L_e|=d_{e}+1$, where $d_{e}$ denotes the number of edges incident to $e$.
\end{theorem}
\vspace{-2pt}

For list-edge-coloring, the previously best known round complexity was $2^{O(\sqrt{\log n})}$, by a classic network decomposition of Panconesi and Srinivasan \cite{panconesi1992improved}, which itself improved on an $2^{O(\sqrt{\log n \cdot \log\log n})}$-round algorithm of Awerbuch et al.\cite{awerbuch1989network}. For low-degree graphs, the best known is an $(\tilde{O}(\sqrt{\Delta}) + O(\log^* n))$-round algorithm of Fraigniaud, Heinrich, and Kosowski \cite{fraigniaud2016local}, which is more general and applies also to ($\Delta+1$)-list-vertex-coloring. There are some known $\poly\log n$-round deterministic edge-coloring algorithms, but all have two major shortcomings: (A) they require more colors, and (B) they are quite restricted and do not work for list-coloring. These algorithms are as follows: (I) a $((2+o(1)) \Delta)$-edge-coloring by Ghaffari and Su\cite{GS17}; (II) a $\Delta\cdot 2^{O\paren{\frac{\log \Delta}{\log\log \Delta}}}$-edge-coloring by Barenboim and Elkin\cite{Barenboim:edge-coloring}---see also \cite[Appendix B]{GS17} for a short proof; and (III) an $O(\Delta\log n)$-edge-coloring by Czygrinow et al.\cite{czygrinow2001coloring}. See also \cite[Chapter 8]{Barenboim:edge-coloring}.

We also note that a recent work of Barenboim, Elkin, and Maimon \cite{barenboim2016edgecoloring} presents an efficient deterministic algorithm for computing a $(\Delta+o(\Delta))$-edge-coloring in graphs with arboricity $a\leq \Delta^{1-\delta}$, for some constant $\delta>0$. In \Cref{lem:edge-coloring-lowArb}, we sketch how a simple combination of the list-edge-coloring algorithm of \Cref{thm:edge-coloring} with $H$-partitionings\cite[Chapter 5.1]{barenboim2013monograph} significantly extends their result.

\paragraph{Improved Randomized Edge-Coloring Algorithm} The deterministic list-edge-coloring algorithm of \Cref{thm:edge-coloring}, in combination with some randomized edge-coloring algorithms of \cite{elkin2015EdgeColoring, barenboim2012locality}, also improves the complexity of randomized algorithms for $(2\Delta-1)$-edge-coloring, making it the first among the four classic problems whose randomized complexity falls down to $\poly(\log\log n)$. 

\begin{corollary}\label{thm:Randedge-coloring} There is a randomized distributed algorithm that computes a $(2\Delta-1)$-edge-coloring in $O(\log^{8} \log n)$ rounds, with high probability, in any $n$-node graph with maximum degree $\Delta$.
\end{corollary}
\vspace{-3pt}
The previous (worst-case) complexity for randomized $(2\Delta-1)$-edge-coloring was $2^{O(\sqrt{\log\log n})}$ rounds\footnote{It is worth noting that the randomized algorithm of \cite{elkin2015EdgeColoring}, as well as its predecessors\cite{panconesi1997randomized, dubhashi1998near}, can also obtain better colorings, even as good as $((1+\eps)\Delta)$-edge-coloring. Though this becomes slow in low-degree graphs.}, due to Elkin, Pettie, and Su\cite{elkin2015EdgeColoring}. By improving this, \Cref{thm:Randedge-coloring} widens the provable gap between the complexity of $(2\Delta-1)$-edge-coloring, which is now in $\poly(\log\log n)$ rounds, and the complexity of maximal matching, which needs $\Omega(\sqrt{\log n/\log\log n})$ rounds\cite{Kuhn:2016:LCL:2906142.2742012}. 

\paragraph{Unified Formulation as Hypergraph Maximal Matching} 
Our first step towards proving \Cref{thm:edge-coloring} is a simple unification of all the aforementioned four classic problems: MIS, ($\Delta+1$)-vertex-coloring, ($2\Delta-1$)-edge-coloring, and maximal matching. We can cast each of these problems as a maximal matching problem  on hypergraphs\footnote{In the \local model, when communicating on a hypergraph, per round each node $v$ can send a message on each of its hyperedges, which then gets delivered to all the other endpoints of that hyperedge. The variant of the model with bounded-size messages can be specialized in a few different ways, see e.g. \cite{kutten2014distributed}. For our purposes, hypergraphs are mainly used for formulating the requirements of the problem, and the real communication happens on the base graph.} of some rank $r$, which depends on the problem, and increases as we move from maximal matching to maximal independent set. In other words, the hypergraph maximal matching problem can be used to obtain a smooth interpolation between maximal matching in graphs and maximal independent set in graphs. Recall that the rank of a hypergraph is the maximum number of vertices in any of its hyperedges. Moreover, a matching in a hypergraph is a set of hyperedges, no two of which share an endpoint. 

We present a reduction from $(2\Delta-1)$-edge-coloring to maximal matching in rank-$3$ hypergraphs, as we sketch next in \Cref{lem:EdgeColoringViaHypergraphMatching}. A similar reduction can be used for list-edge-coloring, as formalized in \Cref{lem:ListEdgeColoringViaHypergraphMatching}. We note that these reductions are inspired by the well-known reduction of Luby from $(\Delta+1)$-vertex-coloring to MIS\cite{luby1986simple, linial1987LOCAL}.

\begin{mdframed}[hidealllines=true, backgroundcolor=gray!20]
\vspace{-9pt}
\begin{lemma}\label{lem:EdgeColoringViaHypergraphMatching} Given a deterministic distributed algorithm $\mathcal{A}$ that computes a maximal matching in $N$-vertex hypergraphs of rank $3$ and maximum degree $d$ in $T(N, d)$ rounds, there is a deterministic distributed algorithm $\mathcal{B}$ that computes a $(2\Delta-1)$-edge-coloring of any $n$-node graph $G=(V, E)$ with maximum degree $\Delta$ in at most $T(3n\Delta, 2\Delta-1)$ rounds.
\end{lemma}
\vspace{-3pt}
\begin{proof}[Proof Sketch]
To edge-color $G=(V, E)$, we generate a hypergraph $H$: Take $2\Delta-1$ copies of $G$. For each edge $e\in E$, let $e_1$ to $e_{2\Delta-1}$ be its copies. For each $e\in E$, add one extra vertex $w_e$ to $H$ and then, change all copy edges $e_1$ to $e_{2\Delta-1}$ to $3$-hyperedges by adding $w_e$ to them. Algorithm $\mathcal{B}$ runs the maximal matching algorithm $\mathcal{A}$ on $H$, and then, for each $e\in E$, if the copy $e_i$ of $e$ is in the computed maximal matching, $\mathcal{B}$ colors $e$ with color $i$. One can verify that each $G$-edge $e$ must have exactly one copy $e_i$ in the maximal matching, and thus we get a $(2\Delta-1)$-edge-coloring.
\end{proof}  
\vspace{-6pt}
\end{mdframed}
\vspace{-5pt}

Besides edge-coloring that gets reduced to hypergraph maximal matching for rank $r=3$, and graph maximal matching which trivially is the special case of $r=2$, we can also formulate MIS and $(\Delta+1)$-vertex-coloring as maximal matching in hypergraphs. For instance, to translate MIS on a graph $G$ to maximal matching on a hypergraph $H$, view each $G$-edge as one $H$-vertex and each $G$-node $v$ as one $H$-edge on the $H$-vertices corresponding to the $G$-edges incident to $v$. However, unfortunately, in this naive formulation, the rank becomes $\Delta$. As such, we do not obtain any improvement over the known algorithms for these problems, in the general case. It remains an intriguing open question whether any alternative formulation, perhaps in combination with other ideas, can help. However, as we shall discuss soon, using some more involved ideas, we obtain improvements for some special cases, which lead to answers for a few other open problems.

\paragraph{Our Hypergraph Maximal Matching Algorithm} Our main technical contribution is an efficient deterministic algorithm for maximal matching in low-rank hypergraphs. In combination with the reduction of \Cref{lem:EdgeColoringViaHypergraphMatching}, this leads to our edge-coloring algorithm stated in \Cref{thm:edge-coloring}.    

\begin{theorem}\label{thm:HypergraphMatching} There is a deterministic distributed algorithm that computes a maximal matching in $O(r^5 \log^{6+\log r} \Delta \cdot \log n)$ rounds\footnote{Throughout this paper, all logarithms are to base $2$.}, in any $n$-node hypergraph with maximum degree $\Delta$ and rank $r$, i.e., where each hyperedge contains at most $r$ vertices.
\end{theorem}
This result has a number of other implications, as we overview in \Cref{subsec:Implications}. Besides those, it also supplies an alternative $\poly\log n$-round deterministic algorithm for maximal matching in graphs, where $r=2$. We remark that $\poly\log n$ deterministic algorithms for graph maximal matching have been known for about two decades, due to the breakthroughs of Ha\'n\'ckowiak, Karonski, and Panconesi\cite{HanckowiakKP98, hanckowiak1999faster}. Moreover, a faster algorithm was recently presented in \cite{FischerGhaffari2017Matching}. However, the methods of \cite{HanckowiakKP98, hanckowiak1999faster, FischerGhaffari2017Matching}, or their natural extensions, inherently rely on rank $r=2$ in a seemingly crucial manner, and they do not extend to hypergraphs of rank $3$ or higher.
The method we develop for \Cref{thm:HypergraphMatching} is quite different and significantly more flexible. We overview this method in \Cref{subsec:method}, and contrast it with the previously known techniques. The difference and the generality of our method becomes more discernible when considering another closely related open problem which did not seem solvable using the methods of \cite{HanckowiakKP98, hanckowiak1999faster} and which can now be solved using a natural, though non-trivial, extension of \Cref{thm:HypergraphMatching}, as we discuss next. 

\paragraph{Extension to MIS in Graphs with Bounded Neighborhood Independence} Consider the problem of computing an MIS in graphs with \emph{neighborhood independence} bounded by an integer $r$, i.e., where the number of mutually non-adjacent neighbors of each node is at most $r$. Notice that maximal matching in graphs is the same as MIS in the corresponding line graph, which is a graph of neighborhood independence $r=2$. It is not clear how to extend the methods of \cite{HanckowiakKP98, hanckowiak1999faster} to MIS in such graphs, even for $r=2$. As an open question alluding to this point, and as ``\emph{a good stepping stone towards the MIS problem in general graphs}", Barenboim and Elkin asked in their book\cite{barenboim2013monograph}: 

\begin{center}
\begin{minipage}{0.95\linewidth}
\begin{mdframed}[hidealllines=true, backgroundcolor=gray!00]
\textsc{Open Problem 11.5}\cite{barenboim2013monograph} $\;$ Devise or rule out a deterministic polylogarithmic algorithm for the MIS problem in graphs with neighborhood independence bounded by $2$. 
\vspace{-5pt}
\end{mdframed}
\end{minipage}
\end{center}
Our method for \Cref{thm:HypergraphMatching} generalizes to MIS in graphs with bounded neighborhood independence, as we state formally next, hence positively answering this open question. 

\begin{theorem}\label{thm:BIMIS} There is a deterministic distributed algorithm that computes a maximal independent set in $O(r^5 \log^{6+\log r} \Delta \cdot \log n)$ rounds, in any $n$-node graph with maximum degree $\Delta$ and neighborhood independence bounded by $r$.
\end{theorem}

Moreover, since Luby's reduction of $(\Delta+1)$-vertex-coloring to MIS\cite{luby1986simple, linial1987LOCAL} increases the neighborhood independence by at most $1$, we also get efficient algorithms for $(\Delta+1)$-vertex-coloring in graphs with bounded neighborhood independence.
\begin{corollary}\label{thm:BIVC} There is a deterministic distributed algorithm that computes a $(\Delta+1)$-vertex-coloring in $O(r^5 \log^{6+\log (r+1)} \Delta \cdot \log n)$ rounds, in any $n$-node graph with maximum degree $\Delta$ and neighborhood independence bounded by $r$. Moreover, the same algorithm solves list-vertex-coloring, where each node $v\in V$ must get a color from an arbitrary given list $L_v$ of colors with $|L_v|\geq\deg(v)+1$.
\end{corollary}

\subsection{Other Implications}
\label{subsec:Implications}
\vspace{-5pt}
Our hypergraph maximal matching algorithm enables us to obtain answers and improvements for some other problems. A family of improvements comes for graph problems in which the main technical challenge is to find a maximal set of ``disjoint" augmenting paths of short length $\ell$. These problems can be phrased as maximal matching in hypergraphs with rank $r=\Theta(\ell)$, essentially by viewing each augmenting path as one hyperedge on its elements (depending on the required disjointness). We next mention the results that we obtain based on this connection.

\paragraph{Maximum Matching Approximation} By integrating our hypergraph maximal matching into the framework of Hopcroft and Karp\cite{HopcroftKarp1973}, we can compute a $(1+\eps)$-approximation of maximum matching in graphs in $\big((\log \Delta/\eps)^{O(\log 1/\eps)}\big)$ rounds. For that, we mainly need to find maximal sets of vertex-disjoint augmenting paths of length at most $\ell=O(1/\eps)$. This is faster than the previously best known deterministic algorithm for $(1+\eps)$-approximation, which required $\log^{O(1/\eps)}n$ rounds\cite{czygrinow2003distributed}. 
We remark that an $O(\log n/\eps^3)$-round randomized $(1+\eps)$-approximation algorithm was presented by Lotker et al.\cite{lotker2008improved}, mainly by computing this maximal set of vertex-disjoint augmenting paths using Luby's randomized MIS algorithm \cite{luby1986simple}.

\paragraph{Low-Out-Degree Orientation and (Pseudo-)Forest Decomposition} By integrating our hypergraph maximal matching into the low-out-degree orientation framework of Ghaffari and Su\cite{GS17}, we can compute orientations with out-degree at most $\lceil (1+\eps)\lambda\rceil$, for any $0<\eps<1$, in graphs with arboricity $\lambda$. For that, we mainly need to find maximal sets of disjoint augmenting paths of length $\ell=O(\log n/\eps)$. This low-out-degree orientation directly implies a decomposition into $\lceil (1+\eps)\lambda\rceil$ edge-disjoint pseudo-forests. For constant $\eps$ and even $\eps=\Omega(1/\poly\log n)$, the round complexity of the resulting algorithm is quasi-polylogarithmic--- that is, $2^{O(\log^2 \log n)}$. Although this is not a polylogarithmic complexity, it gets close and it is almost exponentially faster than the previously best known $2^{O(\sqrt{\log n})}$ deterministic algorithm\cite{GS17, panconesi1992improved}. This improvement can be viewed as partial solution for Open Problem 11.10 of Barenboim and Elkin\cite{barenboim2013monograph}, which asks for an efficient deterministic distributed algorithm for decomposing the graph into less than $2\lambda$ forests.

\subsection{Our Method for Hypergraph Maximal Matching, in a Nutshell}\label{subsec:method}
\vspace{-5pt}
The main ingredient in our results is our hypergraph maximal matching algorithm. Here, we present a brief overview of this algorithm. 

Before the overview, we note that the key technical novelty in our hypergraph maximal matching algorithm is developing an efficient deterministic distributed \emph{rounding} method, which transforms \emph{fractional} hypergraph matchings to \emph{integral} hypergraph matchings. This becomes more instructive when viewed in the context of the recent results of Ghaffari, Kuhn, and Maus\cite{ghaffari2016complexity}, which show that deterministically rounding fractional solutions of certain linear programs to integral solutions while approximately preserving some linear constraints is the \emph{``the only obstacle"} for efficient deterministic distributed algorithms. In other words, if we find an efficient deterministic method for approximately rounding certain linear programs, we would get efficient algorithms for essentially all the classic local graph problems, including MIS. See \cite{ghaffari2016complexity} for the precise statement. 

Our rounding for hypergraph matchings can be seen as a drastic generalization of the rounding methods we presented recently for matching in normal graphs\cite{FischerGhaffari2017Matching}. The methods of \cite{FischerGhaffari2017Matching} do not extend to hypergraphs (even for rank $r=3$), for reasons that we will discuss soon. The new rounding method we present is more general and significantly more flexible. As such, we are hopeful that this deterministic rounding method will prove useful for a wider range of problems, and may potentially serve as a stepping stone towards a $\poly\log n$-time deterministic MIS algorithm. 

We next present a high-level overview of our hypergraph maximal matching algorithm.   

\paragraph{Matching Approximation} The core of our maximal matching algorithm is an algorithm that computes a matching whose size is within an $O(r^3)$-factor of the maximum matching. Once given such an approximation algorithm, we can easily find a maximal matching within $O(r^3\log n)$ iterations, by repeated applications of this approximation algorithm, each time adding the found matching to the output matching, and then removing the found matching and its incident edges from the hypergraph.
 
\paragraph{Fractional Matchings} In approximating the maximum matching, the main challenge is finding an \emph{integral} matching with such an approximation guarantee. Finding a \emph{fractional} matching --- where each edge $e$ has a value $x_e \in [0,1]$ such that for each vertex $v$ we have $\sum_{e\in E(v)} x_{e} \leq 1$ --- with such an approximation is trivial, and can be done in $O(\log \Delta)$ rounds: initially, set $x_e=1/\Delta$ for all edges $e$. Then, for $\log \Delta$ iterations, each time double the values $x_e$ of all the edges $e$ for which all vertices $v\in e$ have $\sum_{v\in E(v)} \leq 1/2$. One can see that this produces a $(2r)$-approximation. The challenge thus is in \emph{rounding} fractional matchings to integral matchings, without losing much in the size. 

\paragraph{Known Rounding Methods for Graphs, and Their Shortcomings} In graphs with rank $r=2$, this rounding can be done essentially with no loss. Indeed, this is the core part of the recent maximal matching algorithm of Fischer and Ghaffari\cite{FischerGhaffari2017Matching}, which finds a maximal matching in $O(\log^2 \Delta \log n)$ rounds. The method of \cite{FischerGhaffari2017Matching} rounds any fractional matching in graphs in $O(\log \Delta)$ iterations, in each iteration moving a $2$-factor closer to integrality while decreasing the matching size only negligibly, by a $(1-\frac{\eps}{\Theta(\log \Delta)})$-factor. Hence, even after all the rounding iterations, the overall loss is a negligible $\eps$-factor, for a desirably small $\eps>0$. Although the algorithms of \cite{HanckowiakKP98, hanckowiak1999faster} are not explicitly phrased in this rounding framework, one can see that the principle behind them is the same. The reader familiar with \cite{HanckowiakKP98, hanckowiak1999faster} might recall that the key component is, roughly speaking, to decompose edges of any regular graph into two groups, say red and blue, so that \emph{almost all} nodes see a fair split of their edges into the two colors. This is a special case of rounding for regular graphs.

This whole methodology of rounding without more than a $o(1)$-factor loss in the size seems to be quite limited, and it certainly gets stuck at rank $r=2$. For the interested reader, we briefly sketch the obstacle: all of those matching rounding methods\cite{FischerGhaffari2017Matching, HanckowiakKP98, hanckowiak1999faster} decompose the edges of the graph into bipartite low-diameter degree-$2$ graphs (i.e., short even-length cycles) --- aside from a smaller portion of some not-so-nice parts, which are handled separately --- and then $2$-color edges of each short cycle so that each node has half of its edges in each color. Then, in rounding, one color is raised by a $2$-factor while the other is dropped to zero. Unfortunately, this type of locally-balanced splittings of edges does not seem within reach for hypergraphs, as of now. Indeed, if we could solve that, we would get far more consequential results: Ghaffari et al.\cite{ghaffari2016complexity} recently proved that this splitting problem for hypergraphs is \emph{`complete'}, meaning that if one can do such a splitting in polylogarithmic time for all hypergraphs, we get polylogarithmic deterministic algorithms for all the classic local problems, including MIS. 

\paragraph{Challenges in Rounding for Hypergraphs} When trying to deterministically round fractional matchings in hypergraphs, we face essentially two challenges: (1) It is not clear how to efficiently perform any slight rounding--- e.g., rounding all fractional values so that the minimum moves from at least $1/d$ to at least $2/d$ without violating the constraints --- without a considerable loss in the matching size. (2) An even more crucial issue comes from the need to do many levels of rounding. Even once we have an efficient solution for a single iteration of rounding, which moves say a constant factor closer to integrality, a $\Theta(r)$-factor reduction of the matching size seems inevitable. However, if we do this repeatedly, and our matching size drops by a $\Theta(r)$-factor in each rounding iteration, the matching size would become too small. Notice that we need about $O(\log \Delta)$ levels of $2$-factor roundings. If we decrease by an $\Omega(r)$-factor per iteration, we would be left with a matching of size a factor $1/r^{\Theta(\log \Delta)} = 1/\poly(\Delta)$ of the maximum matching, which is essentially useless. We next mention a bird's eye view of our rounding method.    

\paragraph{Our Rounding Method for Matchings in Hypergraphs} We devise a rounding procedure for hypergraph matchings which rounds the fractional matching by an $L$-factor--- i.e., raising fractional values by an $L$-factor---while reducing the matching size only by a $\Theta(r)$-factor. On a high level, this rounding is by recursion on $L$. The base level of the recursion is an algorithm that rounds the fractional matching by a constant factor, for $L=O(1)$, with only a $\Theta(r)$-factor decrease of the matching size. This part is somewhat simpler and is performed efficiently using defective coloring results of \cite{kuhn2009weak}. This is a solution for the first challenge above. To overcome the second challenge, our method interleaves some iterations of \emph{rounding} with \emph{refilling} the fractional matching. In particular, suppose that we would like to do an $L$-factor rounding of a given fractional matching $\vec{x}$, thus producing an output fractional matching $\vec{y}$ with fractional values raised by an $L$-factor compared to $\vec{x}$. We do this in $\Theta(r)$ iterations, using a number of $\sqrt{2L}$-factor rounding procedures. Concretely, per iteration, we first `\emph{remove}' the current output fractional matching $\vec{y}$ from the input fractional matching $\vec{x}$, in a sense to be made precise, and then we apply two successive $(\sqrt{2L})$-factor rounding operations on the left-over fractional matching. This creates a fractional matching which is rounded by a factor of $2L$, but may be a $(1/\Theta(r^2))$-factor smaller than $\vec{x}$. We add (a half of) this to the current fractional matching $\vec{y}$, in a sense to be made precise. The removal and also the addition are done carefully, so as to ensure that the size of the output fractional matching grows by about a $(1/\Theta(r^2))$-factor of the size of $\vec{x}$ while the fractionality is by an $L$-factor better than the one of $\vec{x}$. After $\Theta(r)$ such iterations, we get that the output fractional matching is a $\Theta(r)$-approximation of the input. 

\paragraph{Extension to MIS in Graphs with Bounded Neighborhood Independence} When moving from matchings in hypergraphs to independent sets in graphs of neighborhood independence at most $r$, it is not directly clear how to define a fractional solution of an MIS in such graphs. Note that the integrality gap of the natural LP relaxation might be linear in $\Delta$. However, any MIS is within an $r$-factor of a maximum independent set, and this can in fact be generalized to maximal fractional solutions of the following kind. We start by setting the fractional values of all nodes to $0$ and then, we iteratively increment the value of some nodes. As long as right after incrementing the value of a node $v$ the total value in the $1$-neighborhood of $v$ does not exceed $1$, the total value of the resulting fractional solution is guaranteed to be within an $r$-factor of a maximum independent set. We call such a fractional solution a \emph{greedy packing} and show that our rounding scheme for hypergraph matching can be adapted to greedy packings of graphs of bounded neighborhood independence. 

Integral greedy packings are exactly independent sets. Thus, integral greedy packings of the line graph of a hypergraph $H$ correspond to matchings of $H$. However, we note that a fractional greedy packing of the line graph of $H$ is not the same as a fractional matching of $H$. We believe that this stresses the robustness of our approach. For example, when running the MIS algorithm for graphs of bounded neighborhood independence on the line graph of a bounded rank hypergraph $H$, we get a slightly different but equally efficient algorithm for computing a maximal matching of $H$.

\vspace{-3pt}
\section{Maximal Matching and Edge-Coloring in Hypergraphs}
\label{sec:hypergraphmatching}
\vspace{-4pt}

In this section, we present our hypergraph maximal matching algorithm, thus proving \Cref{thm:HypergraphMatching}. Then, at the end \Cref{subsec:EdgeColoring}, we use this hypergraph maximal matching algorithm to prove our edge-coloring results, including \Cref{thm:edge-coloring} and \Cref{thm:Randedge-coloring}.

For our hypergraph maximal matching algorithm, the key part is a \emph{matching approximation} procedure that finds a matching whose size is at least a $(1/(32r^3))$-factor of the maximum matching.

\begin{lemma}\label{lemma:MaxMatchingApprox}
There is a deterministic distributed algorithm that computes a $(32r^3)$-approximate matching in $O\paren{r^2 \log^{6 + \log r} \Delta }$ rounds, given an $O(r^2\Delta^2)$-edge-coloring. 
\end{lemma}

Once we have this approximation algorithm, we can find a maximal matching by iteratively applying this matching approximation procedure to the remainder hypergraph, for $O(r^3\log n)$ iterations, each time removing the found matching and all its incident hyperedges. This is formalized in the proof of \Cref{thm:HypergraphMatching} in \Cref{sec:wrap-up}.

Over the next two subsections, we discuss the matching approximation procedure of \Cref{lemma:MaxMatchingApprox}. We note that finding a \emph{fractional} matching with size close to the maximum matching is straightforward, as we soon overview in \Cref{sec:fractionalApprox}. The challenge is in finding an integral matching with the same guarantee. In other words, the core technical component of our method is an algorithm for \emph{rounding} fractional hypergraph matchings to integral matchings, without losing much in the size. In particular, we present our deterministic rounding technique for hypergraph matchings in \Cref{sec:rounding}.

\subsection{Fractional Matching Approximation}\label{sec:fractionalApprox}

In the following, we present a simple $O(\log \Delta)$-round algorithm that computes a $(2r)$-approximate fractional matching.

\paragraph{Some Notions and Terminology for Fractional Matchings} Given a hypergraph $H=(V, E)$, a fractional matching of $H$ is an assignment of values $\vec{x}\in [0,1]^{|E|}$  to edges such that for each vertex $v\in V$, we have $\sum_{e\in E(v)} x_e \leq 1$. Here, $E(v):=\{e \in E\colon v \in e\}$ is the set of edges incident to $v$. We say a vertex $v$ is \emph{half-tight} in the given fractional matching $\vec{x}$ if $\sum_{e\in E(v)} x_e \geq \frac{1}{2}$. Moreover, we say $\vec{x}$ is a $(1/d)$-fractional matching if each edge $e \in E$ has $x_e\geq 1/d$ or $x_e=0$.

\paragraph{Greedy Fractional Matching Algorithm}
Initially, we set $x_e=\frac{1}{\Delta}$ for all edges $e$. This obviously is a valid fractional matching. Then, for $\log \Delta$ iterations, in each iteration, we freeze all the edges that have at least one half-tight vertex and then raise the value of all unfrozen edges by a 2-factor.

This way, we always keep a valid fractional matching, since only the values of edges incident to non-half-tight vertices are increased. Moreover, within $O(\log \Delta)$ iterations all edges will be frozen. We next show that this property already implies an approximation ratio $2r$. 

\begin{lemma}\label{lemma:halftight-2rapprox}The greedy algorithm described above computes a $(2r)$-approximate fractional matching. Moreover, any (fractional) matching $\vec{x}$ with the property that each edge has at least one half-tight endpoint is a $(2r)$-approximation.
\end{lemma}
\begin{proof}
We show that $\vec{x}$ must have size at least a $(1/(2r))$-factor of a maximum matching $M^*$ employing an argument based on counting in two ways. To that end, we give $1$ dollar to each edge $e \in M^*$ and ask it to redistribute this money among edges in such a way that no edge $e'$ receives more than $2rx_{e'}$ dollars. 
This can be achieved as follows. 
Each edge $e \in M^*$ asks a half-tight vertex, say $v\in e$, to distribute $e$'s dollar on $e$'s behalf. Vertex $v$ does so by splitting this money among its incident edges $e'\in E(v)$ proportionally to the edge values $x_{e'}$. In this way, every edge $e' \in E(v)$ receives no more than $2x_{e'}$ dollars from $v$. This is because $v$ is half-tight and because it cannot have more than one incident edge in $M^*$, hence does not receive more than $1$ dollar. Since an edge can receive money only from its vertices, every edge $e'$ receives at most $2r x_{e'}$ dollars in total.  
\end{proof}

\subsection{Rounding Fractional Matchings in Hypergraphs}\label{sec:rounding}
Our method for rounding fractional matchings is recursive, and parametrized mainly by a parameter $L$ which captures the extent of the performed rounding. In simple words, given a fractional matching $\vec{x}\in [0,1]^{|E|}$, the method $\round(\vec{x}, L)$ rounds $\vec{x}$ by an $L$-factor. That is, if in the input fractional matching $\vec{x}$ the smallest (non-zero) value is $1/d$, then in the output fractional matching the smallest (non-zero) value is at least $L/d$. On the other hand, the guarantee is that the output fractional matching has size at least a $(1/(4r))$-factor of the input fractional matching. The functionality of this rounding method is abstracted by the following definition.

\begin{definition}[$L$-factor rounding] 
Given a $(1/d)$-fractional matching $\vec{x} \in [0, 1]^{|E|}$ --- i.e. where for each $e\in E$, we have $x_{e}\geq 1/d$ or $x_e=0$ --- the method $\round(\vec{x}, L)$ computes an $(L/d)$-fractional matching $\vec{y} \in [0, 1]^{|E|}$ such that $\sum_{e\in E} y_{e} \geq \frac{1}{4r} \sum_{e\in E} x_{e}$. 
\end{definition}

\begin{remark}\label{rmrk:fractionality} The method requires some condition on the values of $L$ and $d$. Since $L/d$ refers to the fractionality, the statement is meaningful only when $L \leq d$. Due to some small technicalities, we will perform the recursive parts of rounding only for values of $L$ that satisfy a slightly stronger condition of $L \log^2 L \leq d$. For the remaining cases, we resort to our basic rounding. 
\end{remark}

We explain our rounding method in two main parts. The first part, explained in \Cref{sec:basicRounding}, is a procedure that we use as the base case, to round the matching by a constant factor $L=O(1)$ in $O(r^2 + \log \Delta)$ rounds. The second part, discussed in \Cref{sec:recursiveRounding}, is the recursive step which explains how our $L$-factor rounding works by making a few calls to $\sqrt{2L}$-factor rounding procedures, and a few smaller steps. Finally, in \Cref{sec:wrap-up}, we combine these rounding procedures with the previously seen algorithm of \Cref{sec:fractionalApprox} for fractional matchings to obtain our matching approximation procedure of \Cref{lemma:MaxMatchingApprox}. 

\subsubsection{Basic Rounding}\label{sec:basicRounding}
In this subsection, we explain our base case rounding procedure for small rounding parameters, i.e., $L=O(1)$. Throughout, we will assume that the base hypergraph already has an $O(r^2 \Delta^2)$-edge-coloring, which can be computed easily using Linial's algorithm \cite{linial1987LOCAL}, in $O(\log^* n)$ rounds.

\begin{lemma}[Basic Rounding]\label{lemma:hypergraphbasicRounding} 
There is an $O(L^2r^2 + \log \Delta)$-round deterministic distributed algorithm that turns a $(1/d)$-fractional matching $\vec{x}$ into an $(L/d)$-fractional matching $\vec{y}$ with $\sum_{e\in E} y_e \geq \frac{1}{2r}\sum_{e\in E} x_e$, for any $L\leq d$. 
\end{lemma}

\subsubsection*{Algorithm Outline and Intuitive Discussions}
Let $E_{x}$ be the set of all edges $e$ for which $x_{e} > 0$, and let $H_{x}=(V, E_{x})$ be the subgraph of $H$ with this edge set. Notice that $H_{x}$ has degree at most $d$, because $\vec{x}$ is a $(1/d)$-fractional matching. Our goal is to compute a fractional matching $\vec{y}$, supported on the edge set $E_x$, such that for each edge $e \in E_x$, at least one of its endpoints $v\in e$ is half-tight in $\vec{y}$, meaning that $\sum_{e'\in E_x(v)} y_{e'}\geq 1/2$. One can easily see that such a fractional matching is a $(2r)$-approximation of $\vec{x}$, i.e., $\sum_{e\in E} y_e \geq \frac{1}{2r}\sum_{e\in E} x_e$. Thus, the goal is to find a fractional matching $\vec{y}$ such that for each edge $e\in E_x$, at least one of its endpoints is \emph{half-tight} in $\vec{y}$. Furthermore, we want $\vec{y}$ to be $(L/d)$-fractional, meaning that all the non-zero $y_e$-values must be greater than or equal to $L/d$. 

If we had no concern for the time complexity, we could go through the color classes of edges one by one, each time setting $y_e=1$ for all edges of that color, and then removing edges of $E_x$ that have half-tight vertices. This would ensure that, at the end, all edges in $E_x$ have at least one half-tight endpoint. However, this would require time proportional to the number of colors. Even if we were given an ideal edge-coloring for free, that would be $\Omega(d)$ rounds, which is too slow for us. 

To speed up the process, we use a relaxed notion of edge-coloring, namely \emph{defective edge-coloring}, which allows us to have much less colors, while each color class has a bounded number of edges incident to each vertex, say $k$. Now, we cannot raise the $y_e$-values of all the edges of the same color at the same time to $y_e=1$, because that would be too fast and could violate the condition $\sum_{e' \in E_x(v)} y_{e'} \leq 1$. However, we can raise each of these edge values to say $y_e = \frac{1}{2k}$ and still be sure that the summation $\sum_{e'\in E_x(v)} y_{e'}$ for each node does not increase faster than an additive $1/2$. That is because there are only $k$ edges incident to each node, per color class. If we freeze and remove all edges that now have one half-tight vertex, these fractional value raises would never violate the condition $\sum_{e'\in E_x(v)} y_{e'} \leq 1$, thus always lead to a valid fractional matching. 

\subsubsection*{The Basic Rounding Algorithm}
To materialize the above intuitive approach, we first compute a \emph{defective edge-coloring} with $O\paren{r^2\Delta^2}$ colors and defect $k= d/(2L)$. Then, we go through the colors, one by one, applying the above fractional-value increases. This ensures that all the non-zero fractional values $y_e$ are at least $\frac{1}{2k} \geq L/d$. At the very end, we perform $O(\log {d/L})$ doubling steps to ensure that each edge has at least one half-tight endpoint. We next explain the steps of this algorithm, and then provide the related analysis. 

\medskip
\paragraph{Part I, Defective Edge-Coloring Algorithm}
We compute a defective edge-coloring of $H_x$ with $O(L^2 r^2)$ colors and defect --- that is, maximum degree induced by edges of the same color --- at most $d/(2L)$, as follows. Let $F=(V_F, E_F)$ be the line graph of $H_x$, that is, the graph which has a vertex $v_e \in V_F$ for every edge $e\in E_x$ and an edge $\{v_e, v_{e'}\}\in E_F$ if $e$ and $e'$ are incident, thus $e \cap e' \neq \emptyset$. Note that $F$ has maximum degree at most $ r \cdot d$, since $H_x$'s maximum degree is bounded by $d$. With the defective coloring algorithm of Kuhn \cite{kuhn2009weak}, we can compute a $(d/(2L)-1)$-defective vertex-coloring of $F$ with $O\big((\frac{r\cdot d}{(d/(2L)-1)})^2\big) =O\paren{L^2r^2}$ colors\footnote{If we happen to have $d/(2L) \leq 1$, then $(d/(2L)-1)$-defective coloring becomes a degenerate case of the definition, as $(d/(2L)-1) \leq 0$, and then by convention, this simply means proper coloring. In that case the algorithm of Kuhn \cite{kuhn2009weak} provides a proper coloring with $O(d^2 r^2) = O(L^2r^2) $ colors.}. Exploiting the given $O\paren{r^2 \Delta^2}$-edge-coloring of $H$, and thus $H_x$, which is an $O\paren{r^2\Delta^2}$-vertex-coloring of the line graph $F$, we can make this algorithm run in $O(\log^*(r \Delta))$ rounds. The vertex-coloring of the line graph with defect $d/(2L) -1$ is an edge-coloring of $H_x$ where every edge has at most $d/(2L)-1$ incident edges of the same color, resulting in at most $d/(2L)$ many edges of the same color incident to each vertex. 

\medskip
\paragraph{Part II, Fractional Matching Computation via Defective Coloring} We process the colors of the $d/(2L)$-defect-defective coloring one by one, in $O\paren{L^2 r^2}$ iterations. In the $i^{th}$ iteration, for each non-frozen edge $e$ with color $i$, we raise $y_e$ from $y_e=0$ to $y_e=L/d$. Then for each node $v$ that is already half-tight, meaning that $\sum_{e\in E_x(v)} y_e\geq 1/2$, we freeze all the edges incident to $v$. This means the fractional value of these edges will not be raised in the future. Notice that since we raise values only incident to nodes that are not already half-tight, and as for each such node the summation goes up by at most $\frac{d}{2L} \cdot \frac{L}{d} =1/2$, the vector $\vec{y}$ always remains a fractional matching, meaning that we always have $\sum_{e\in E_x(v)} y_e \leq 1$ for each node $v$.

At the very end, once we are done with processing all colors, some edges in $E_x$ may remain without any half-tight endpoint. Though any such edge $e$ would itself have $y_e=L/d$. We perform $\log{(d/L)}$ iterations of doubling, where in each iteration, we double all the fractional values $y_e$ for all edges that do not have a half-tight endpoint. At the end, we are ensured that each edge has at least one half-tight endpoint, and moreover, each non-zero fractional value $y_e$ is at least $L/d$. 
  
\begin{lemma}\label{lemma:E_y-rounding}
The above algorithm computes an $(L/d)$-fractional matching $\vec{y}$ such that $\sum_{e\in E} y_e \geq \frac{1}{2r}\sum_{e\in E} x_e$, in $O(L^2 r^2 + \log (d/L) + \log^*(r\Delta)) = O(L^2r^2 + \log \Delta)$ rounds.
\end{lemma}
\begin{proof}
The round complexity of the algorithm comes from the $O(\log^*(r\Delta))$ rounds spent for computing the defective edge-coloring, $O(L^2 r^2)$ rounds for processing the colors of the defective coloring one by one, and then $O(\log (d/L))$ rounds for the final doubling steps. 

It is clear by construction that the computed vector $\vec{y}$ is a fractional matching, because we always have $\sum_{e\in E_x(v)} y_e \leq 1$, and that it is $(L/d)$-fractional, because the smallest non-zero $y_e$ value that we use is $L/d$. What remains to be proved is that $\sum_{e\in E} y_e \geq \frac{1}{2r}\sum_{e\in E} x_e$. For that, we use the property that the fractional matching $\vec{y}$ that we compute is such that for each $e\in E_x$, at least one of the vertices $v\in e$ must be half-tight, meaning that $\sum_{e\in E(v)} y_e \geq 1/2$. We use this property to argue that the fractional matching $\vec{y}$ has size at least a $(1/(2r))$-factor of $\vec{x}$. This is done via a blaming argument along the same lines as the proof of \Cref{lemma:halftight-2rapprox}. We let every edge $e \in E_x$ put $x_e$ dollars on edges $e' \in E_y$ as follows. Each edge $e$ passes its $x_e$ dollars to one of its half-tight vertices $v\in e$. Then, the half-tight vertex $v$ distributes these $x_e$ dollars among all its incident edges $e'\in E_x(v)$ proportionally to the values $y_{e'}$. As $\vec{x}$ is a fractional matching, in this way, $v$ cannot receive more than $1$ dollar in total from its incident edges in $E_x$. Therefore, and since $v$ is half-tight, no edge $e'$ incident to $v$ receives more than $2y_{e'}$ dollars from $v\in e'$. In total, an edge $e'\in E_y$ can receive at most $2ry_{e'}$ dollars from edges in $E_x$, at most $2y_{e'}$ from each of its endpoints. Therefore, $\sum_{e \in E} x_e \leq 2r \sum_{e\in E} y_e$.  
\end{proof}

\subsubsection{Recursive Rounding}\label{sec:recursiveRounding}

We explain a recursive method $\round(\vec{x}, L)$ that given a ($1/d$)-fractional matching $\vec{x}$ computes an ($L/d$)-fractional matching $\vec{y}$ such that $\sum_{e\in E} y_{e} \geq \frac{1}{4r} \sum_{e\in E} x_{e}$. This procedure will be applied when $L$ is greater than some fixed constant. The procedure works mainly by a number of recursive calls to $\sqrt{2L}$-factor rounding procedures, and a few additional steps. 

\begin{lemma}[Recursive Rounding]\label{lemma:recursiveRounding}
There is an $O\paren{ (r^2 + \log \Delta)\log^{5 + \log r} L}$-round deterministic distributed algorithm that turns a $(1/d)$-fractional matching $\vec{x}$ into an $(L/d)$-fractional matching $\vec{y}$ with $\sum_{e\in E} y_e \geq \frac{1}{4r}\sum_{e\in E} x_e$, for any $L$ such that $L\log^2 L\leq d$. 
\end{lemma} 

\paragraph{The Recursive Rounding Algorithm}
The method $\round(\vec{x}, L)$ consists of $16r$ iterations. Initially, we set $y_e=0$ for all edges. Then, in $16$ iterations, we gradually grow $\vec{y}$ while keeping it $(L/d)$-fractional. 

The process in each iteration is as follows: 
\begin{itemize}
\item We first generate a fractional matching $\vec{z}$ by initially setting it equal to $\vec{x}$, and then removing from it each edge $e$ that is incident to a at least one half-tight vertex of $\vec{y}$. In other words, for each vertex $v$ such that $\sum_{e\in E(v)} y_{e} \geq 1/2$, we set $z_{e} = 0$ for all $e$ with $v \in e$; for all other edges, we set $z_{e}=x_{e}$. 
\item We first perform $\mathsf{round}(\vec{z}, \sqrt{2L})$, producing some intermediate $(\sqrt{2L}/d)$-fractional matching $\vec{z'}$. Then we perform $\mathsf{round}(\vec{z'}, \sqrt{2L})$. This creates a $(2L/d)$-fractional matching $\vec{z''}$ whose size is at least a factor $(\frac{1}{4r})\cdot (\frac{1}{4r}) = \frac{1}{16r^2}$ of the size of $\vec{z}$. 
\item We divide the values of this fractional matching $\vec{z''}$ by a $2$-factor, creating an $(L/d)$-fractional matching, and we add the result to $\vec{y}$. Thus, we effectively update $\vec{y} \gets \vec{y}+\vec{z''}/2$.
\end{itemize} 

\begin{remark}
Recall the promise from \Cref{rmrk:fractionality} that we will apply the rounding method only for values such that ${L\log^2 L}\leq d$. The main reason for this stronger condition, compared to the more natural condition of $L\leq d$, is the $2$-factor that we have in the recursive rounding call. For instance, the matching $\vec{z''}$ is a $(2L/d)$-fractional matching and thus, for this to be meaningful, we need ${2L}\leq d$. However, with the stronger condition that ${L\log^2 L}\leq d$, we can say that the promise is satisfied throughout the recursive calls. For instance, in the second call to $\round(\vec{z'}, \sqrt{2L})$, the new condition would be ${\sqrt{2L} (\log \sqrt{2L})^2} \leq {d/\sqrt{2L}}$, which is readily satisfied given that ${L\log^2 L}\leq {d}$ and $L\geq 8$.
\end{remark}

\paragraph{Analysis of the Recursive Rounding} We next provide the related analysis. In particular, \Cref{lemma:valid} proves that the generated fractional matching $\vec{y}$ is valid, \Cref{lemma:y-to-x} proves that it is a good approximation of $\vec{x}$, and \Cref{lemma:recursiveComplexity} analyzes the running time of this recursive procedure.
\medskip
 
\begin{lemma}\label{lemma:valid}
The fractional matching $\vec{y}$ is valid, meaning that $\sum_{e\in E(v)} y_e \leq 1$ for all vertices $v$. 
\end{lemma}
\begin{proof}
We show by induction on $i$ that the fractional matching $\vec{y}$ in iteration $i$ does not violate the constraints $\sum_{e\in E(v)} y_e \leq 1$ for all $v$. 
At the beginning, the condition is trivially satisfied. 
If $v$ is half-tight at the beginning of an iteration, then $z_e=0$ and hence $z''_e=0$ for all $e\in E(v)$, thus no value is added to $\sum_{e\in E(v)} y_e$ in this iteration.
If $v$ is not half-tight at the beginning of an iteration, we add at most half of a fractional matching to edges incident to $v$, thus at most a value $1/2$ to the summation $\sum_{e\in E(v)} y_e$. More formally, we have $\sum_{e\in E(v)} z''_e \leq 1$, thus $\sum_{e \in E(v)} z''_e/2 \leq \frac{1}{2}$, which results in a new value of at most $\sum_{e\in E(v)} (y_e + z''_e /2) \leq 1$. 
\end{proof}

\begin{lemma}\label{lemma:y-to-x} At the end of $16r$ iterations, we have $\sum_{e\in E} y_{e} \geq \frac{1}{4r} \sum_{e\in E} x_{e}$.
\end{lemma}
\begin{proof}
Consider one iteration and suppose that $\sum_{e\in E} y_{e} \leq \frac{1}{4r} \sum_{e\in E} x_{e}$. We first show that then $\sum_{e\in E} z_{e} \geq \frac{1}{2} \sum_{e\in E} x_{e}$ by a blaming argument along the same lines as the proof of \Cref{lemma:halftight-2rapprox}. For that, we let every edge $e\in E$ which is incident to a half-tight vertex in $\vec{y}$ put $x_e$ dollars on edges in a manner that each edge $e'$ receives at most $2r y_{e'}$ dollars. This can be done by sending those $x_{e}$ dollars of $e$ to (one of) its $\vec{y}$-half-tight vertex $v$, and then letting $v$ distribute these $x_e$ dollars among its incident edges $e'\in E(v)$ proportionally to the values $y_{e'}$. Since $\vec{x}$ is a matching, each vertex $v$ in total receives at most $1$ dollar from its incident edges. Then, since $v$ is $\vec{y}$-half-tight, it can distribute this dollar among its incident edges such that no edge $e'$ receives more than $2y_{e'}$ dollars from one of its endpoints $v$. Now, an edge $e'\in E$ can possibly receive $2y_{e'}$ dollars from each of its (half-tight) endpoint vertices, thus in total at most $2ry_{e'}$ dollars. Therefore, indeed $$\sum_{e \in E \colon \exists v \in e \colon \sum_{e' \in E(v)} y_{e'} \geq 1/2} x_{e} \leq 2r \sum_{e \in E} y_e \leq \frac{2r}{4r}\sum_{e\in E} x_e=\frac{1}{2}\sum_{e\in E} x_e.$$It follows that if $\sum_{e\in E} y_{e} \leq \frac{1}{4r} \sum_{e\in E} x_{e}$, then $\sum_{e\in E}z_e \geq \frac{1}{2}\sum_{e \in E} x_e$. 
Thus, in each such iteration, the matching $\vec{y}$ grows by at least $$\sum_{e\in E} z''_{e}/2\geq \frac{1}{2}\cdot \frac{1}{16r^2}\sum_{e \in E} z_e \geq  \frac{1}{2} \cdot \frac{1}{16r^2} \cdot \frac{1}{2} \sum_{e\in E} x_e.$$ Therefore, after at most $16r^2$ iterations, we have $\sum_{e\in E} y_{e} \geq \frac{1}{4r} \sum_{e\in E} x_{e}$.
\end{proof}

\begin{lemma}\label{lemma:recursiveComplexity}The algorithm $\round(\vec{x}, L)$ has round complexity $O((r^2 + \log \Delta) \log^{5+\log r} L)$.
\end{lemma}
\begin{proof}
The complexity $R(L)$ of the rounding algorithm $\round(\vec{x}, L)$ follows the recursive inequality $R(L) \leq 16 r (R(\sqrt{2L}) + R(\sqrt{2L}) + O(1))$. Furthermore, we have the base case solution of $R(L) = O(L^2 r^2 + \log \Delta)$ for $L=O(1)$. The claim can now be proved by an induction on $L$, as formalized in \Cref{lem:recursion}. Here, instead of the formal calculations, we mention an intuitive explanation: the complexity gets multiplied by roughly $32r$ as we move from $L$ to $\sqrt{2L}$. There are, roughly, $\log\log L$ such moves, and hence the complexity gets multiplied by $(32r)^{\log\log L} < \log^{5+\log r} L$ until we reach the base case of $L=O(1)$, where the base complexity is $O(r^2 + \log \Delta)$ by \Cref{lemma:hypergraphbasicRounding}.
\end{proof}
\begin{remark}\label{rmrk:constant} We remark that we have not tried to optimize the constant that appears in the exponent of the round complexity $O((r^2 + \log \Delta) \log^{5+\log r} L)$. This constant mainly comes from the constant in the number of iterations in our recursive rounding, which is currently set to $16 r$, for simplicity. Optimizing this constant may be of interest specially for small values of $r$. In particular, for the case of $r=3$, which leads to our edge-coloring result via \Cref{lem:EdgeColoringViaHypergraphMatching}, the current constants lead to a complexity $O(\log^{6+\log_2{3}} \Delta) = O(\log^{7.6} \Delta)$ for $\Theta(1)$-approximation of maximum matching, thus an $O(\log^{7.6} \Delta \log n)$ complexity for maximal rank-$3$ matching and also edge-coloring. One can easily improve this to $O(\log^{6.75} \Delta \log n)$, and perhaps even further, by adjusting the constants throughout.
\end{remark}

\subsection{Wrap-up}\label{sec:wrap-up}
We now use our rounding procedure to find the approximate maximum matching of \Cref{lemma:MaxMatchingApprox}. 

\begin{proof}[Proof of \Cref{lemma:MaxMatchingApprox}]
First, we compute a $({1}/{\Delta})$-fractional $(2r)$-approximate matching $\vec{x}$ in $O(\log \Delta)$ rounds, by the greedy algorithm described in \Cref{sec:fractionalApprox}. Then, we apply the recursive rounding from \Cref{lemma:recursiveRounding} for $L = \Delta/\log^2 \Delta$. This produces a $(1/\log^2 \Delta)$-fractional matching $\vec{x'}$ whose size is a $(1/(8r^2))$-factor of the maximum fractional matching of the hypergraph, in $O\paren{\log^{5 + \log r}\Delta(r^2 + \log \Delta)}$ rounds. To finish up the rounding, we apply the basic rounding of \Cref{lemma:hypergraphbasicRounding} for $L=\log^2 \Delta$, which runs in $O(r^2\log^2 \Delta)$ and produces a $1$-fractional --- i.e. integral --- matching $\vec{x''}$ whose size is at least a $(1/(4r))$-factor of the size of $\vec{x'}$. Hence, the final produced integral matching is a $(32r^3)$-approximation of the maximum matching.
\end{proof}

We compute a maximal matching via iterative applications of this matching approximation.
\begin{proof}[Proof of \Cref{thm:HypergraphMatching}]
First, we pre-compute an $O(r^2 \Delta^2)$-edge-coloring of $H$ in $O(\log^*n)$ rounds by Linial's algorithm \cite{linial1987LOCAL}. Then, iteratively, we apply the maximum matching approximation procedure of \Cref{lemma:MaxMatchingApprox} to the remaining hypergraph. We add the found matching $M$ to the matching that we will output at the end, and remove $M$ along with its incident edges from the hypergraph. 

In each iteration, the size of the maximum matching of the remaining hypergraph goes down to at least a factor of $1-1/(32r^3)$ of the previous size. This is because otherwise we could combine the matching computed so far with the maximum matching in the remainder hypergraph to obtain a matching larger than the maximum matching in $H$. After $O(r^3 \log n)$ repetitions, the remaining maximum matching size is $0$, which means the remaining hypergraph is empty. Hence, we have found a maximal matching in $O\left(\log^*n + r^3 \log n \left(r^2 \log^{6+ \log r} \Delta \right)\right) = O(r^5 \log^{6 + \log r} \Delta \cdot \log n)$ rounds.
\end{proof}

\subsection{Edge-Coloring}\label{subsec:EdgeColoring}

\begin{lemma}\label{lem:ListEdgeColoringViaHypergraphMatching} Given a deterministic distributed algorithm $\mathcal{A}$ that computes a maximal matching in $N$-vertex hypergraphs of rank $3$ and maximum degree $d$ in $T(N, d)$ rounds, there is a deterministic distributed algorithm $\mathcal{B}$ that solves list-edge-coloring of any $n$-node graph $G=(V, E)$ with maximum degree $\Delta$ in at most $T(2n\Delta^2, 2\Delta-1)$ rounds. In the list-edge-coloring problem, each edge $e$ must choose its color from  an arbitrary given list $L_e$ of colors with $|L_e|=d_{e}+1$, where $d_{e}$ denotes the number of edges adjacent to $e$.
\end{lemma}
\vspace{-5pt}
\begin{proof}[Proof Sketch]
To edge-color $G=(V, E)$, we generate a hypergraph $H$: For each edge $e\in E$ with list-color $L_e$, we take $|L_e|$ copies of $e=\{v, u\}$ as follows: For each color $i \in L_e$, we take one copy $e_i$ of $e$ which is put incident to copies $v_i$ and $u_i$ of $v$ and $u$. Thus, if two adjacent edges $e=\{v, u\}$ and $e'=\{v, u'\}$ have a common color $i \in L_{e} \cap L_{e'}$, then their $i^{th}$ copies $e_i$ and $e'_i$ will be present and will both be incident to $v_i$. Notice that for each vertex $v$, at most $2\Delta^2$ copies of it will be used because for each of the edges $e$ incident to $v$, at most $|L_{e}| < 2\Delta$ additional copies of $v$ are added.

Algorithm $\mathcal{B}$ runs the maximal matching algorithm $\mathcal{A}$ on $H$, and then, for each edge $e\in E$, if the copy $e_i$ of $e$ is in the computed maximal matching, $\mathcal{B}$ colors $e$ with color $i$. One can verify that each $G$-edge $e$ has exactly one copy $e_i$ in the maximal matching, and thus we get a list-edge-coloring.
\end{proof}  

\medskip
\noindent
\textbf{\Cref{thm:edge-coloring}.} \emph{There is a deterministic distributed algorithm that computes a $(2\Delta-1)$-edge-coloring in $O(\log^{7} \Delta \cdot \log n)$ rounds, in any $n$-node graph with maximum degree $\Delta$. Moreover, the same algorithm solves list-edge-coloring, where each edge $e\in E$ must get a color from a given list $L_e$ of colors with $|L_e|=d_{e}+1$, where $d_{e}$ denotes the number of edges adjacent to $e$.}
\medskip
\begin{proof}
Follows directly from \Cref{lem:ListEdgeColoringViaHypergraphMatching} and \Cref{thm:HypergraphMatching} (with the slightly improved exponent constant, noted in \Cref{rmrk:constant}).
\end{proof}

\begin{remark}
We note that \Cref{lem:ListEdgeColoringViaHypergraphMatching}, and hence also \Cref{thm:edge-coloring}, can be easily extended from graphs to hypergraphs. In particular, list-edge-coloring of hypergraphs of rank $r$ can be reduced to maximal matching in hypergraphs of rank $r+1$. Thus, we can obtain a deterministic list-edge-coloring algorithm for hypergraphs of rank $r$ with round complexity $O(r^5 \log^{6+ \log (r+1)} \Delta \log n)$ rounds.  
\end{remark}

As stated before, the deterministic list-edge-coloring algorithm of \Cref{thm:edge-coloring}, in combination with known randomized algroithms of Elkin, Pettie, and Su \cite{elkin2015EdgeColoring} and Johansson~\cite{johansson1999simple}, leads to a $\poly(\log\log n)$ randomized algorithm for $(2\Delta-1)$-edge-coloring.

\medskip
\noindent\textbf{\Cref{thm:Randedge-coloring}.} \emph{There is a randomized distributed algorithm that computes a $(2\Delta-1)$-edge-coloring in $O(\log^{8} \log n)$ rounds, with high probability.} 
\medskip

\begin{proof}[Proof of \Cref{thm:Randedge-coloring}]
For $\Delta=\Omega(\log^2n)$, we run the $O\paren{\log^*\Delta+ \frac{\log n}{\Delta^{1-o(1)}}}$-round algorithm by Elkin, Pettie, and Su \cite{elkin2015EdgeColoring} for $((1+\eps)\Delta)$-edge-coloring, in $O(\log^*\Delta)$ rounds.
For $\Delta=o(\log^2n)$, we first apply the simple randomized coloring algorithm of Johansson~\cite{johansson1999simple} for $O(\log \Delta) = O(\log\log n)$ rounds. In particular, in each iteration, every remaining edge $e$ independently picks a color $q_{e}$ from its remaining palette uniformly at random. If there is no incident edge that picked the same color $q_e$, then the edge $e$ is colored with this color $q_{e}$ and removed from the graph. Moreover, the color $q_e$ gets deleted from the palettes of every incident edge. As proved in e.g. \cite{barenboim2012locality}, after $O(\log \Delta)$ rounds, this procedure leaves us with a graph where each connected component of remaining edges has size at most $N=\poly \log n$. On these components, we then run the list-edge-coloring algorithm of \Cref{thm:edge-coloring} to complete the partial coloring. This takes at most $O(\log^{8} N) = O(\log^8 \log n)$ rounds. Hence, including the $O(\log\log n)$ initial rounds, the overall complexity is $O(\log^8 \log n)$ rounds.
\end{proof}

\section{MIS in Graphs of Bounded Neighborhood Independence}
\label{sec:boundedindep}

We will now generalize the hypergraph maximal matching algorithm of
\Cref{sec:hypergraphmatching} to computing maximal independent sets
and $(\Delta+1)$-vertex colorings of graphs of \emph{bounded neighborhood
  independence}, thus proving \Cref{thm:BIMIS,thm:BIVC}. Formally, the neighborhood independence of a graph is
defined as follows.

\begin{definition}[Bounded Neighborhood Independence]\label{def:BNI}
  For an integer $r\geq 1$, we say that a graph $G=(V,E)$ has
  \emph{neighborhood independence} at most $r$ if for every node
  $v\in V$, the graph $G[N(v)]$ induced by the set $N(v)$ of neighbors
  of $v$ has independence number at most $r$. 
\end{definition}

We note that the line graph of a hypergraph $H$ of rank $r$ has
neighborhood independence at most $r$. This is because for each
hyperedge $e$ of $H$, all incident hyperedges $f$ share at least
one node with $e$ and because all the hyperedges sharing a node of
$H$ form a clique in the line graph of $H$. Hence, the neighborhood of
each edge can be covered by at most $r$ cliques in the line
graph. Despite the fact that graphs of neighborhood independence $r$
are significantly more general than line graphs of rank-$r$
hypergraphs, we show that our techniques for computing a maximal
matching in rank-$r$ hypergraphs can be generalized to computing an
MIS in graphs of neighborhood independence $r$, and this, in fact, even with exactly the same asymptotic dependency on $r$ and
$\Delta$.

While in the case of hypergraph matchings, the natural LP relaxation
leads to fractional solutions that are within a small factor of a
maximum matching, it is not as straightforward to model fractional
versions of independent sets in graphs of bounded neighborhood
independence in a meaningful way. Note that for example even for
$r=O(1)$, the integrality gap of the natural LP relaxation of the
maximum independent set problem might be as large as $\Omega(\Delta)$.
In the following, we study special fractional solutions $\vec{x}$ that
assign a non-negative value $x_v\geq 0$ to each node $v\in V$ of a
graph $G=(V,E)$ and allow to approximate maximum and maximal
independent sets in $G$ if $G$ if a graph of bounded neighborhood
independence. For convenience, we first introduce some
notation. Recall that given a graph $G=(V,E)$ and a node $v\in V$, we
use $N(v)$ to denote the set of neighbors of $v$. Further, we define
\[
N^+(v) := \set{v} \cup N(v)
\]
to denote the set of nodes in the $1$-neighborhood of $v$. Moreover, 
for node vector $\vec{x}$ assigning values $x_v$ to every
node $v\in V$, for each node $v\in V$, we define
\[
\Sigma_{\vec{x}}(v) := \sum_{u\in N^+(v)} x_u
\]
to be the local sum of the values $x_u$ in the $1$-neighborhood of
$v$. As a fractional relaxation of the independent set of a graph $G$, we define
a \emph{greedy packing} as follows.

\begin{definition}[Greedy Packing]\label{def:greedypacking}
  For a graph $G=(V,E)$, a node vector
  $\vec{x}$ assigning a non-negative value $x_v\geq 0$ to each node
  $v\in V$ is called a \emph{greedy packing} if there
  exists a global order $\prec$ on the nodes $V$ such that
  \[
  \forall v\in V: \quad x_v + \!\!\sum_{u\in N(v) : u\prec
    v}\!\!\!\!\!\!  x_u\ \leq 1.
  \]
\end{definition}

Hence, in a greedy packing, the values $x_v$ can be assigned
to the nodes in some order such that for all nodes $v\in V$, when node
$v$ gets assigned value $x_v$, the sum of the values in $v$'s
$1$-neighborhood is bounded by $1$. 

Analogously to the matching algorithm in
\Cref{sec:hypergraphmatching}, the key part is a recursive algorithm
that finds and independent set that is an approximation of a maximum
independent set in graphs of bounded neighborhood independence. We
formally prove the following result.

\begin{lemma}\label{lemma:MaxISApprox}
  In graphs of neighborhood independence at most $r$, there is a
  deterministic distributeqd algorithm that computes a
  $(32r^3)$-approximate independent set in
  $O\paren{r^2 \log^{6 + \log r} \Delta }$ rounds, given an
  $O(\Delta^2)$-vertex-coloring.
\end{lemma}

We first show that in graphs of bounded neighborhood independence for
any such greedy packing $\vec{x}$, the local sum $\Sigma_{\vec{x}}(v)$
is bounded for all nodes.

\begin{lemma}\label{lemma:greedypacking}
  Let $G=(V,E)$ be a graph with neighborhood independence at most $r$
  and assume that we are given a greedy packing $\vec{x}$. Then, for
  all $v\in V$, we have $\Sigma_{\vec{x}}(v)\leq r$.
\end{lemma}
\begin{proof}
  Consider an arbitrary node $v\in
  V$ and let $G_v$ be the subgraph of $G$
  induced by the nodes in $N^+(v)$. Let $\prec$ be the global order on $V$ which is defined
  by \Cref{def:greedypacking} because $\vec{x}$ is a greedy
  packing. Assume that the nodes in $N^+(v)$ are named
  $u_0,\dots,u_{d(v)}$ such that for all $0\leq i< d(v)$, $u_i \succ
  u_{i+1}$. We construct an MIS $S$ of $G_v$ by processing the nodes
  in $N^+(v)$ in the order $u_0,u_1,\dots, u_{d(v)}$, always adding
  the current node $u_i$ to $S$ if no neighbor of $u_i$ has already
  been added to $S$. In this way, every node $u_i\in N^+(v)\setminus S$ has an
  MIS neighbor $u_j\in N^+(v)$ for which $j < i$ and thus $u_i \prec u_j$. We charge the value
  $x_{u_i}$ of every node $u_i\in N^+(v)\setminus S$ to some MIS
  neighbor $u_j$ for which $j<i$. In addition, the value $x_{u_j}$ of
  each MIS node $u_j\in S$ is charged to the node itself. For each MIS
  node $u_j\in S$, let $X_{u_j}$ be the total value charged to
  $u_j$. We can upper bound $X_{u_j}$ as follows: 
  \[
  X_{u_j} \leq \sum_{u_i \in N^+(u_j)\cap N^+(v) : i > j} x_{u_i} \leq
  x_{u_j} + \sum_{w\in N(u_j): w \prec u_j} x_w\leq 1.
  \]
  The last inequality follows because $\vec{x}$ is a greedy packing
  w.r.t.\ the global order $\prec$. The claim of the lemma now follows
  because $G$ has neighborhood independence bounded by $r$, and thus
  the MIS $S$ can contain at most $|S|\leq r$ nodes.
\end{proof}

We will show how to recursively compute a large greedy packing. Before
doing this, we first prove some useful simple properties of greedy
packings.

\begin{lemma}\label{lemma:greedybasic}
  Let $G=(V,E)$ be a graph, and assume that we are given a global
  order $\prec$ on $V$ and a greedy packing $\vec{x}$ w.r.t.\ the
  order $\prec$. Then, the following statements hold:
  \begin{enumerate}[(1)]
  \item Let $v\in V$ be node for which $\Sigma_{\vec{x}}(v)\leq 1$ and let
    $\vec{y}$ be a node vector such that $y_u = x_u$ for all $u\neq v$
    and such that $y_v \leq x_v + 1 - \Sigma_{\vec{x}}(v)$. Then
    $\vec{y}$ is also a greedy packing.
  \item Let $U\subseteq V$ be a subset of the nodes and let $\vec{y}$
    be a node vector such that $y_v=x_v$ for all nodes
    $v\in V\setminus U$ and such that $y_u\geq x_u$ for all $u\in U$.
    If $\Sigma_{\vec{y}}(u)\leq 1$ for all nodes $u\in U$, $\vec{y}$
    is also a greedy packing.
  \item Let $U\subseteq V$ be a set of nodes $u$ for which
    $\Sigma_{\vec{x}}(u)\leq 1/2$ and consider a node vector $\vec{y}$
    such that $y_v=x_v$ for all $v\not\in U$ and such that
    $y_u = 2x_u$ for all $u\in U$. Then $\vec{y}$ is also a greedy
    packing.
  \end{enumerate}
\end{lemma}
\begin{proof}
  We first prove claim (1). Consider a global order $\prec_0$ that is
  obtained from $\prec$ by moving node $v$ to the very end of the
  order (without changing the relative order of any of the other
  nodes). We claim that $\vec{y}$ is a greedy packing w.r.t.\ the
  global order $\prec_0$. For all nodes $u\neq v$, the condition of
  \Cref{def:greedypacking} follows because $\vec{x}$ is a greedy
  packing w.r.t.\ the global order $\prec$. For node $v$, the
  condition follows because
  $\Sigma_{\vec{y}}(v) = \Sigma_{\vec{x}}(v) + y_v - x_ v \leq 1$.

  Claim (2) follows from claim (1) by sequentially processing the
  nodes in $U$. We start with vector $\vec{x}$ and when processing
  node $u$, we replace the current value $x_u$ of node $u$ by
  $y_u$. Because for all nodes $y_v\geq x_v$, for each of the
  intermediate vectors $\vec{z}$, we have
  $\Sigma_{\vec{z}}(u)\leq 1$ for all $u\in U$. The conditions
  for claim (1) are therefore satisfied for each node $u\in U$.

  Finally, to prove claim (3), observe that because we have
  $\Sigma_{\vec{x}}(u)\leq 1/2$ for all nodes $u\in U$, the
  local sum for the nodes in $u$ is still bounded by $1$ even if
  we double the values of all nodes $v\in V$. Claim (3) therefore
  follows as a special case of claim (2).
\end{proof}

In order to recursively compute a large greedy packing, we will need
to be able to add a new greedy packing to an existing one. The
next lemma shows in which way this can be done. In the following,
given a real-valued non-negative node vector $\vec{x}$ and a parameter
$c>0$, we call a node $v\in V$ $c$-saturated if
$\Sigma_{\vec{x}}(v) \geq c$.

\begin{lemma}\label{lemma:greedycombine}
  Let $G=(V,E)$ be a graph and let $\vec{x}$ be a greedy packing of
  $G$. Further, let $F\subseteq V$ be the nodes of $G$ that are not
  $1/2$-saturated w.r.t.\ $\vec{x}$ and let $\vec{y}$ be a greedy
  packing of $G$ for which $y_v>0$ only for $v\in F$. Then, the
  fractional assignment $\vec{z} := \vec{x} + \vec{y}/2$ is a greedy
  packing of $G$.
\end{lemma}
\begin{proof}
  Assume that $\vec{x}$ is a greedy packing of $G$ w.r.t.\
  to the global order $\prec_x$ on $V$ and that $\vec{y}$ is a greedy
  packing of $G$ w.r.t.\ to the global order $\prec_y$ on
  $F$. Note that for all nodes $v\in F$, we have
  $\Sigma_{\vec{x}}(v)< 1/2$. Therefore for the greedy packing
  $\vec{x}$ the condition of \Cref{def:greedypacking} is satisfied for
  the nodes in $F$ for every choice of the global order
  $\prec_x$. W.l.o.g., we can therefore assume that $\prec_x$ first
  orders all nodes in $V\setminus F$ and it then orders the nodes in
  $F$ in an arbitrary way. We can thus define a global order $\prec$
  on the nodes $V$ as a combination of $\prec_x$ and $\prec_y$ in an
  obvious way. The order $\prec$ first orders the nodes in
  $V\setminus F$ in the same order as $\prec_x$ and it then orders the
  nodes in $F$ in the same order as $\prec_y$. We show that
  $\vec{z} = \vec{x} + \vec{y}/2$ is a greed packing of $G$
  w.r.t.\ to the global order $\prec$. For each node
  $v\in V\setminus F$, the condition of \Cref{def:greedypacking} is
  satisfied because $\vec{x}$ is a greedy packing. For a node
  $v\in F$, we have
  \[
  z_v + \sum_{u \in N(v) : u \prec v} z_u \ =\ 
  \frac{y_v}{2} + \sum_{u\in N(v)\cap F: u \prec_y v} \frac{y_u}{2} +
  x_v + \sum_{u\in N(v) : u \prec_x v} x_u\ <\ 
  \frac{1}{2} + \frac{1}{2}\ =\  1,
  \]
  and therefore the claim of the lemma follows.
\end{proof}

As in the case of computing matchings in hypergraphs, our goal is to
start with a large fractional greedy packing and to gradually round
the fractional solution to an integer one of approximately the same
size. For a given parameter $\delta>0$, we call a non-negative
real-valued node vector $\vec{x}$ $\delta$-fractional if for every
node $v\in V$, either $x_v=0$ or $x_v\geq\delta$. Given a parameter
$L>1$, we show how to recursively turn a $\delta$-fractional greedy
packing into an $(L\delta)$-fractional greedy packing of a similar
size. The proof follows the same basic structure as the rounding for
fractional hypergraph matchings in
\Cref{sec:basicRounding,sec:recursiveRounding}. The following lemma
provides a way to upper bound the size of a greedy packing
in terms of another greedy packing. We will use it to
compare the size of a computed $(L\delta)$-fractional greedy packing to
the existing $\delta$-fractional greedy packing.

\begin{lemma}\label{lemma:comparepackings}
  Let $\vec{x}$ and $\vec{y}$ be two
  greedy packings of a graph $G=(V,E)$ with neighborhood independence at most $r$. Further,
  let $U\subseteq V$ be the set of nodes of $V$ for which
  $\Sigma_{\vec{y}}(v)\geq 1/2$. We have
  \[
  \sum_{v\in V} y_v \geq \frac{1}{2r}\sum_{v\in U} x_v.
  \]
\end{lemma}
\begin{proof}
  To prove the lemma, we use a blaming argument. Let $V_y$ be the set
  of nodes for which $y_v>0$. We distribute all the $x_v$-values of
  nodes in $U$ among the nodes in $V_y$. That is, for each node
  $v\in V_y$, we define a variable $\alpha_v$ such that
  $\sum_{v\in V_y} \alpha_v = \sum_{v\in U} x_v$. More concretely, we define the values
  $\alpha_v$ for each node $v\in V_y$ as follows:
  \begin{equation}\label{eq:basicsize}
    \alpha_v := \sum_{u\in N^+(v)\cap U} \!\!\! x_u \cdot
    \frac{y_v}{\Sigma_{\vec{y}}(u)}\ \leq\ 
    \sum_{u\in N^+(v)\cap U} \!\!\! x_u\cdot 2y_v\ =\
    2 y_v\cdot \Sigma_{\vec{x}}(v).
  \end{equation}
  Hence, every node $u\in U$ distributes its value $x_u$ among the
  neighboring nodes in $v\in V_y$ proportionally to the values
  $y_v$. Because $\vec{x}$ is a greedy packing of $G$,
  \Cref{lemma:greedypacking} implies that
  $\Sigma_{\vec{x}}(v)\leq r$ for all $v\in V$. Together
  with \Cref{eq:basicsize}, we thus get $\alpha_v\leq 2r y_v$ for all
  $v\in V_y$ and the claim of the lemma follows.
\end{proof}

\subsection{Basic Rounding of Greedy
  Packings}\label{sec:BNIbasicRounding}

\begin{lemma}[Basic Rounding of Greedy
  Packings]\label{lemma:BNIbasicRounding}
  Let $G=(V,E)$ be a graph with neighborhood independence at most
  $r\geq 1$. Further, assume that a parameter $L>1$, an integer
  $d\geq L$, and a $(1/d)$-fractional greedy packing $\vec{x}$ of
  $G$ are given. If in addition, an $O(\Delta^2)$-vertex-coloring of
  $G$ is given, there is an
  $O((rL)^2 + \log^*\Delta + \log d)$-round deterministic
  distributed algorithm that computes an $(L/d)$-fractional
  greedy packing $\vec{y}$ for which $y_v>0$ only if $x_v>0$ and
  such that $\vec{y}$ is of size
  \[
  \sum_{v\in V} y_v \geq \frac{1}{2r} \sum_{v\in V} x_v.
  \]
\end{lemma}
\begin{proof}
  Let $V_x$ be the set of nodes $v\in V$ for which $x_v>0$ and let
  $G_x=G[V_x]$ be the subgraph of $G$ induced by $V_x$. Note that
  because $\vec{x}$ is a greedy packing,
  \Cref{lemma:greedypacking} implies that for every node $v\in V$,
  $\Sigma_{\vec{x}}(v)\leq r$ and since $\vec{x}$ is
  $(1/d)$-fractional, this implies that $G_x$ has maximum degree at
  most $r\cdot d$. We compute the $(L/d)$-fractional greedy
  packing $\vec{y}$ in two steps. In a first step, we compute an
  arbitrary $(L/d)$-fractional greedy packing $\vec{z}$ of $G$ by
  assigning value $z_v=L/d$ to a subset of the nodes $v\in V_x$. In
  the second step, we obtain $\vec{y}$ from $\vec{z}$ by iteratively
  doubling the value of each node that is not $(1/2)$-saturated at
  most $O(\log d)$ times.

  For the first step, we apply the deterministic defective coloring
  algorithm of Kuhn \cite{kuhn2009weak}. For a $C$-vertex-colored graph $G$
  of maximum degree $\Delta$ and a parameter $p\geq 1$, the algorithm
  allows to compute a $p$-defective $O((\Delta/p)^2)$-coloring of $G$
  in time $O(\log^*C)$. That is, the algorithm assigns one of
  $O((\Delta/p)^2)$ colors to each node of $G$ such that the subgraph
  induced by each of the colors has maximum degree at most $p$. We
  apply the defective coloring algorithm of \cite{kuhn2009weak} to the
  graph $G_x$ with parameter $p=d / (2L)$. Because we are given an
  $O(\Delta^2)$-coloring of $G$ (and thus also of $G_x$), the time for
  computing this defective coloring is $O(\log^*\Delta)$ and because
  the maximum degree of $G_x$ is at most $dr$, the number of colors of
  the defective coloring is at most $O((rL)^2)$.

  We now compute an initial $(L/d)$-fractional greedy packing
  $\vec{z}$ as follows. For all nodes $v\in V\setminus V_x$, we set
  $z_v=0$. For the nodes in $V_x$, we iterate through the
  $O((rL)^2)$ colors of the defective coloring of $G_x$ and 
	process all nodes of the same color in parallel. At the beginning,
  we set $z_v=0$ for all $v\in V_x$. When processing the nodes of
  colors $c$, for each node $v\in V_x$ of color $c$, we set $z_v=L/d$
  if and only if $\Sigma_{\vec{z}}(v) \leq 1/2$. Because each
  node of color $c$ has at most $d /(2L)$ neighbors of color $c$,
  this implies that even after this step, $\Sigma_{\vec{z}}(v) \leq 1$
  for all nodes of color $c$. Claim (2) of \Cref{lemma:greedybasic}
  therefore implies that throughout this process, vector $\vec{z}$
  remains a valid greedy packing. Because at the end all
  non-zero values of $\vec{z}$ are equal to $L/d$, clearly, $\vec{z}$
  is $(L/d)$-fractional. Note also that for all nodes $v\in V_x$ for
  which $z_v=0$, we have $\Sigma_{\vec{z}}(v)\geq 1/2$.

  To obtain the greedy packing $\vec{y}$ from $\vec{z}$ we
  first set $\vec{y}=\vec{z}$ and we then proceed in synchronous
  rounds. Let $V_y\subseteq V_x$ be the set of nodes for which
  $z_v>0$. In each round, each node $v\in V_y$ for which
  $\Sigma_{\vec{y}}(v)\leq 1/2$ doubles its value $y_v$. The process
  stops when $\Sigma_{\vec{y}}(v)\geq 1/2$ for all nodes $v\in V_y$.
  Because $y_v\geq 1/2$ implies that $\Sigma_{\vec{y}}(v)\geq 1/2$,
  this happens after at most $\log(d/L)\leq \log d$ rounds. Claim (3)
  of \Cref{lemma:greedybasic} implies that the vector $\vec{y}$
  remains a valid greedy packing throughout this process. 

  We therefore obtain an $(L/d)$-fractional greedy packing
  $\vec{y}$ where for each node $v\in V_x$, we have
  $\Sigma_{\vec{y}}(v)\geq 1/2$. \Cref{lemma:comparepackings} then
  shows that that
  $\sum_{v\in V} y_v \geq 1/(2r)\sum_{v\in V} x_v$, which
  concludes the proof.
\end{proof}

\subsection{Recursive Rounding of Greedy
  Packings}\label{sec:BNIrecursiveRounding}

We next explain a recursive method $\round(\vec{x},L)$ that given a
$\delta$-fractional greedy packing $\vec{x}$ of a graph $G=(V,E)$ of neighborhood independence $\leq r$
computes an $(L\delta)$-fractional greedy packing $\vec{y}$ of $G$
of size $\sum_{v\in V}y_v\geq \frac{1}{4r}\sum_{v\in V} x_v$ and
such that $y_v>0$ only if $x_v>0$. The method $\round(\vec{x},L)$ runs
in $16r$ phases. At the beginning, $\vec{y}=0$ and in each phase, some
values of $\vec{y}$ are increased. As soon as for some node $v\in V$,
$\Sigma_{\vec{y}}(v)\geq 1/2$, the value $y_v$ is not increased any
further. In each phase, the method therefore first defines a
$\delta$-fractional greedy packing $\vec{z}$ which is identical to
$\vec{x}$ on all nodes $v$ for which $\Sigma_{\vec{y}}(v)<1/2$ and which is
$0$ on all other nodes. On this vector $\vec{z}$, the method is called
recursively with parameter $\sqrt{2L}$, resulting in a
$(\sqrt{2L}\delta)$-fractional greedy packing
$\vec{z'}$. Afterwards, the method is again called recursively with
parameter $\sqrt{2L}$ on the vector $\vec{z'}$, resulting in a
$(2L\delta)$-fractional greedy packing $\vec{z''}$. Finally, the
vector $\vec{y}$ is updated by adding $\vec{z''}/2$ to it. The
algorithm is also summarized in \Cref{alg:recrounding}.

\begin{algorithm}[t]
  \caption{Recursive Rounding Algorithm $\round(\vec{x},L)$}
\label{alg:recrounding}
\SetKwInOut{Input}{Input}\SetKwInOut{Output}{Output}

\Input{A graph $G=(V,E)$ with a $\delta$-fractional greedy
  packing $\vec{x}$, a parameter $L<1/2\delta$}
\Output{A $\delta L$-fractional packing $\vec{y}$ of $G$}
\BlankLine
\uIf{$L\leq 4$}{run the basic rounding algorithm of
  \Cref{lemma:BNIbasicRounding}}
\Else{
  $\vec{y} := \vec{0}$\;
  \For {phase $1,\ldots,16r$}{
    $V_z := \set{v\in V : \Sigma_{\vec{y}}(v) < 1/2}$\;
    define $\vec{z}$ s.t.\ $z_v = x_v$ for $v\in V_z$ and $z_v=0$
    otherwise\;
    $\vec{z'} := \round(\vec{z}, \sqrt{2L})$\;
    $\vec{z''} := \round(\vec{z'}, \sqrt{2L})$\;
    $\vec{y} := \vec{y} + \vec{z''}/2$
  }
}
\end{algorithm}

The following lemma shows that the algorithm $\round(\vec{x}, L)$
computes an $(L\delta)$-fractional greedy packing of size within a
factor $4r$ of the size of $\vec{x}$.

\begin{lemma}\label{lemma:BNIrecursiveRounding}
  Assume that we are given a graph $G=(V,E)$ with neighborhood independence at most $r$, parameters $0<\delta<1$ and
  $L<1/(2\delta)$, and a $\delta$-fractional greedy packing
  $\vec{x}$ of $G$. The algorithm $\round(\vec{x},L)$ computes an
  $(L\delta)$-fractional greedy packing $\vec{y}$ for which $y_v>0$
  only if $x_v>0$ and such that
  \begin{equation}\label{eq:recursiveapprox}
  \sum_{v\in V}y_v \geq \frac{1}{4r}\sum_{v\in V} x_v.
  \end{equation}
\end{lemma}
\begin{proof}
  Note that for $L\leq 4$, the algorithm directly applies the basic
  rounding algorithm of \Cref{lemma:BNIbasicRounding} and the claims of
  the lemma thus directly hold by applying
  \Cref{lemma:BNIbasicRounding}. Let us therefore assume that $L>4$ and
  let us therefore (inductively) also assume that the recursive calls
  to $\round(\vec{z}, \sqrt{2L})$ and $\round(\vec{z'}, \sqrt{2L})$
  satisfy the claims of the lemma.

  We first show that $\vec{y}$ is an $(L\delta)$-fractional greedy
  packing of $G$ and that $y_v>0$ only if $x_v>0$. Note that
  $z_v>0$ only if $x_v>0$ and we have $z_v'>0$ only if $z_v>0$ and
  $z_v''>0$ only if $z_v'>0$ because that is guaranteed by the
  recursive calls to $\round(\vec{z}, \sqrt{2L})$ and
  $\round(\vec{z'}, \sqrt{2L})$. Because $y_v$ is only increased for
  nodes $v\in V$ for which $z_v''>0$, we therefore have $y_v>0$ only
  if $x_v>0$ throughout the algorithm. To see that $\vec{y}$ is
  $(L\delta)$-fractional, note that because $\vec{x}$ is
  $\delta$-fractional, the recursive calls to
  $\round(\vec{z},\sqrt{2L})$ and $\round(\vec{z'},\sqrt{2L})$
  guarantee that $\vec{z'}$ is $(\sqrt{2L}\delta)$-fractional and
  $\vec{z''}$ is $(2L\delta)$-fractional. We update $\vec{y}$ by adding
  $\vec{z''}/2$ and thus an $(L\delta)$-fractional vector to it. Thus,
  $\vec{y}$ is $(L\delta)$-fractional at all times during the execution of
  \Cref{alg:recrounding}. We prove that $\vec{y}$ at all times is a
  greedy packing by induction on the number of
  phases. Clearly at the beginning when $\vec{y}=\vec{0}$, $\vec{y}$
  is a greedy packing. Also, whenever, $\vec{y}$ is
  updated, we add $\vec{z''}/2$ to it. Note that because $\vec{z''}$
  is the result of the call to $\round(\vec{z'},\sqrt{2L})$,
  $\vec{z''}$ is a greedy packing. Further, $z_v''>0$
  only where $z_v>0$ and thus only for nodes $v$ where
  $\Sigma_{\vec{y}}(v)<1/2$ at the beginning of the respective
  phase. It therefore follows directly from \Cref{lemma:greedycombine}
  that $\vec{y} + \vec{z''}/2$ is a greedy packing of
  $G$.

  It thus remains to show \Cref{eq:recursiveapprox}. As long as
  \Cref{eq:recursiveapprox} does not hold, at the beginning of each of
  the $16r$ phases, we have
  $\sum_{v\in V} x_v > 4r\sum_{v\in V}y_v$. As in
  \Cref{alg:recrounding}, let $V_z$ be the set of nodes for which
  $\Sigma_{\vec{y}}(v) < 1/2$ and $\bar{V_z}:=V\setminus V_z$ be
  the set of nodes for which $\Sigma_{\vec{y}}(v) \geq 1/2$. From
  \Cref{lemma:comparepackings}, we get that
  $\sum_{v\in \bar{V_z}}x_v \leq 2r\sum_{v\in V} y_v$ and we thus
  have
  $\sum_{v\in V}z_v = \sum_{v\in V_z}x_v > \frac{1}{2}\sum_{v\in
    V} x_v$.
  From the guarantees of the recursive calls to
  $\round(\vec{z},\sqrt{2L})$ and $\round(\vec{z'},\sqrt{2L})$, the
  size of $\vec{z''}/2$ that we add to $\vec{y}$ is thus
  \[
  \frac{1}{2}\sum_{v\in V} z_v''\ \geq\ 
  \frac{1}{8r}\sum_{v\in V} z_v'\ \geq\ 
  \frac{1}{32r^2}\sum_{v\in V} z_v\ \geq\
  \frac{1}{64r^2} \sum_{v\in V} x_v.
  \]
  After $16r$ phases, we therefore have
  $\sum_{v\in V} y_v \geq \frac{16r}{64r^2}\sum_{v\in V} x_v$.
\end{proof}

\begin{lemma}\label{lemma:BNIrecursiveComplexity} 
  The algorithm $\round(\vec{x}, L)$ has round complexity $O((r^2 + \log \Delta) \log^{5+\log r} L)$.
\end{lemma}
\begin{proof}
  The proof is identical to the proof of the analogous
  \Cref{lemma:recursiveComplexity} in the hypergraph matching
  analysis.  The complexity $R(L)$ of the rounding algorithm
  $\round(\vec{x}, L)$ follows the recursive inequality
  $R(L) \leq 16 r (R(\sqrt{2L}) + R(\sqrt{2L}) + O(1))$. Furthermore,
  we have the base case solution of $R(L) = O(L^2 r^2 + \log \Delta)$
  for $L=O(1)$. The claim can now be proved by an induction on $L$, as
  formalized in \Cref{lem:recursion}.
\end{proof}

\subsection{Wrap-up}\label{sec:wrap-upBNI}

We now use our rounding procedure to find the approximate independent
set of \Cref{lemma:MaxISApprox}.

\begin{proof}[Proof of \Cref{lemma:MaxISApprox}]
  Let $S^*$ be some maximum independent set of the given graph
  $G=(V,E)$ with neighborhood independence $\leq r$.  We first compute
  a $({1}/{\Delta})$-fractional greedy packing $\vec{x}$ of size at
  least $\frac{1}{2r}\cdot |S^*|$. To compute $\vec{x}$, we initially
  set $x_v=1/\Delta$ for all nodes $v\in V$. As this guarantees that
  $\Sigma_{\vec{x}}(v)\leq 1$ for all $v\in V$, this initial vector
  $\vec{x}$ clearly is a greedy packing. Now, we proceed in
  $\log \Delta$ synchronous rounds, where in each round, all nodes
  $v\in V$ for which $\Sigma_{\vec{x}}(v)<1/2$ double their value
  $x_v$. Claim (3) of \Cref{lemma:greedybasic} implies that the vector
  $\vec{x}$ remains a greedy packing throughout this process. Further,
  after at most $\log\Delta$ doubling steps, we certainly have
  $\Sigma_{\vec{x}}(v)\geq 1/2$ for all nodes $v\in
  V$.
  \Cref{lemma:comparepackings} therefore implies that
  $\sum_{v\in V} x_v \geq \frac{1}{2r}\sum_{v\in V} y_v$ for
  every greedy packing $\vec{y}$ of $G$. The claim that
  $\sum_{v\in V} x_v \geq |S^*|/(2r)$ now follows because for any
  independent $S$ set of $G$, setting $y_v=1$ for $v\in S$ and $y_v=0$
  otherwise results in a greedy packing $\vec{y}$.

  Given the greedy packing $\vec{x}$, we now apply the recursive
  rounding from \Cref{lemma:BNIrecursiveRounding} for
  $L = \Delta/\log^2 \Delta$.  This produces a
  $(1/\log^2 \Delta)$-fractional greedy packing $\vec{x'}$, in
  $O\paren{\log^{5 + \log r}\Delta(r^2 + \log \Delta)}$ rounds, whose
  size is a $(1/(8r^2))$-factor of $|S^*|$. To finish up the rounding,
  we apply the basic rounding of \Cref{lemma:BNIbasicRounding} for
  $L=\log^2 \Delta$, which runs in $O(r^2\log^2 \Delta)$ and produces
  a $1$-fractional --- i.e., integral --- greedy packing $\vec{x''}$
  of size is at least a $(1/(4r))$ times the size of $\vec{x'}$. Hence,
  the final produced integral greedy packing is a
  $(32r^3)$-approximation of the maximum independent set $S^*$.
\end{proof}

We compute a maximal independent set via iterative applications of
this independent set approximation.
\begin{proof}[Proof of \Cref{thm:BIMIS}]
  First, we pre-compute an $O(\Delta^2)$-vertex-coloring of $G$ in
  $O(\log^*n)$ rounds by Linial's algorithm
  \cite{linial1987LOCAL}. Then, iteratively, we apply the maximum
  independent set approximation procedure of \Cref{lemma:MaxISApprox}
  to the remaining graph. We add the found independent set $S$ to the
  independent set that we will output at the end, and remove $S$ along
  with its neighbors from the graph.

  In each iteration, the size of the maximum independent set of the
  remaining graph goes down to at least a factor of
  $1-1/(32r^3)$. This is because otherwise we could combine the
  independent set computed so far with the maximum independent set in
  the remaining graph to obtain an independent set larger than the
  maximum independent set in $G$. After at most $O(r^3 \log n)$ repetitions,
  the remaining independent set size is $0$, which means the
  remaining graph is empty. Hence, we have found a maximal
  independent set in 
  $O\left(\log^*n + r^3 \log n \left(r^2 \log^{6+ \log r} \Delta
    \right)\right) = O(r^5 \log^{6 + \log r} \Delta \cdot \log n)$
  rounds.
\end{proof}

\section{Other Implications: Matching Approximation, and Orientations}
\subsection{Approximating Maximum Matching in Graphs}
\begin{theorem} There is a deterministic distributed algorithm that computes a $(1+\eps)$-approximation of maximum matching in $O(\poly(\frac{1}{\eps})\cdot (\frac{1}{\eps} \log \Delta)^{7+\log 1/\eps})$ rounds, for any $\eps\in (0, 1]$.
\end{theorem}
\begin{proof}
We first discuss an algorithm with complexity $O(\poly(\frac{1}{\eps})\cdot (\frac{1}{\eps} \log \Delta)^{6+\log 1/\eps} \cdot \log n)$, and then explain how a small change improves the complexity to $O(\poly(\frac{1}{\eps})\cdot (\frac{1}{\eps} \log \Delta)^{7+\log 1/\eps})$.

We follow a well-known approach of Hopcroft and Karp\cite{HopcroftKarp1973} of increasing the size of the matching using short augmenting paths.
Given a matching $M$, an augmenting path $P$ with respect to $M$ is a path that starts with an unmatched vertex, then alternates between non-matching and matching edges, and ends in an unmatched vertex. Augmenting the matching $M$ with this path $P$ means replacing the matching edges in $P \cap M$ with the edges $P\setminus M$. Notice that the result is a matching, with one more edge.

The approximation algorithm variant of Hopcroft and Karp\cite{HopcroftKarp1973} works as follows: For each $\ell =1$ to $2(1/\eps)-1$, we find a maximal set of vertex-disjoint augmenting paths of length $\ell$, and we augment them all. Hopcroft and Karp \cite{HopcroftKarp1973} show that this produces a $(1+\eps)$-approximation of maximum matching. See also \cite{lotkerMatchingImproved}, where they use the same method to obtain a $O(\log n/\eps^3)$-round randomized distributed algorithm for $(1+\eps)$-approximation of maximum matching, using the help of the $O(\log n)$ round randomized MIS algorithm of Luby \cite{luby1986simple}.

What remains to be discussed is how do we compute a maximal set of vertex-disjoint augmenting paths of a given length $\ell \leq 2(1/\eps)-1$. This can be easily formulated as a hypergraph maximal matching for a hypergraph of rank at most $1/\eps + 1$: create a hypergraph $H$ by including one vertex for each unmatched node and also one vertex for each matching edge. Then, each augmenting path is simply a hyperedge made of its elements, i.e., its unmatched vertices and its matching edges. This hypergraph has rank at most $1/\eps + 1$, maximum degree at most $\Delta^{2(1/\eps)}$, and the number of its vertices is no more than $n$. Moreover, a single round of communication on this hypergraph can be simulated in $O(1/\eps)$ rounds of the base graph, simply because each hyperedge spans a path of length at most $O(1/\eps)$.
Hence, we can directly apply \Cref{thm:HypergraphMatching} to compute a maximal matching of it, i.e., a maximal set of vertex-disjoint augmenting paths. This runs in $O(\frac{1}{\eps^6}(\frac{2}{\eps} \log \Delta)^{5+\log (1/\eps +1 )} \cdot \log n)$ rounds. This is the complexity of the algorithm for each one value of $\ell \in [1, 2/\eps -1]$. Thus, the overall complexity is at most $O(\frac{1}{\eps^7}(\frac{2}{\eps} \log \Delta)^{5+\log (1/\eps +1)} \cdot \log n)$. 

What remains to be discussed is removing the $\log n$ factor from the complexity. We next provide a sketch. In the above algorithm, we compute a maximal set of disjoint augmenting paths, and this precise maximality necessitates the $\log n$-factor (in our approach). However, we do not need such a precise maximality. It suffices if the set of disjoint augmenting paths is almost maximal, in particular in the sense that the fraction of the remaining augmenting paths is less than $\poly(\eps\Delta^{-1/\eps})$, say. Then, even if we permanently remove all nodes that have such a remaining augmenting path, we lose only a negligible $\poly({\eps}{\Delta})$-factor of the matching, which at the end only changes our approximation ratio to $1+2\eps$. To compute such an almost maximal set of disjoint augmenting paths, instead of $O(r^3 \log n)$ iterations in the proof of \Cref{thm:HypergraphMatching}, it suffices to have $O\big(r^3 \log(\poly(\Delta^{1/\eps}/\eps))\big)$ iterations. This brings down the overall complexity to 
$O(\poly(\frac{1}{\eps})\cdot (\frac{1}{\eps} \log \Delta)^{7+\log 1/\eps})$.
\end{proof}

\subsection{Orientations with Small Out-Degree}\label{subsec:orientation}
\begin{theorem} There is a deterministic distributed algorithm that computes an orientation with maximum out-degree at most  $\lceil\lambda (1+\eps)\rceil$ in $2^{O(\log^2 (\log n/\eps))}$ rounds, for any $\eps>0$, in any graph with arboricity at most $\lambda$.
\end{theorem}
\begin{proof}
We follow the approach of Ghaffari and Su\cite{GS17}, which iteratively improves the orientation, i.e., reduces its maximum out-degree, using suitably defined augmenting paths. They developed this approach and used it along with Luby's randomized MIS algorithm \cite{luby1986simple} to obtain a polylogarithmic round randomized algorithm for finding an orientation with out-degree at most $\lceil\lambda (1+\eps)\rceil$. We show how to turn that algorithm into a quasi-polylogarithmic round deterministic algorithm, mainly by replacing their MIS module with our hypergraph maximal matching algorithm.

Let $D=\lceil\lambda (1+\eps)\rceil$. Given an arbitrary orientation, we call a path $P$ an augmenting path for this orientation if $P$ is a directed path that starts in a node with out-degree at least $D+1$ and ends in a node with out-degree at most $D-1$. Augmenting this path means reversing the direction of all of its edges. Notice that this would improve the orientation, as it would decrease the out-degree of one of the nodes whose outdegree is above the budget $D$, without creating a new such node.

Let $G_0$ be the graph with our initial arbitrary orientation. Define $G'_{0}$ to be a directed graph obtained by adding a source node $s$ and a sink node $t$ to $G_0$. Then, we add $outdeg_{G_0}(u)-D$ edges from $s$ to every node $u$ with outdegree at least $D+1$, and $D-outdeg_{G_0}(u)$ edges from every node $u$ with outdegree at most $D-1$ to $t$. We will improve the orientation gradually, in $\ell=O(\log n/\eps)$ iterations. In the $i^{th}$ iteration, we find a maximal set of edge-disjoint augmenting paths of length $3+i$ from $s$ to $t$ in $G'_i$, and then we reverse all these augmenting paths. The resulting graph is called $G'_{i+1}$. 

Ghaffari and Su\cite[Lemma D.6]{GS17} showed that in this manner, each time the length of the augmenting path increases by at least an additive $1$. Moreover, they showed that at the end of the process, no augmenting paths of length at most $\ell=O(\log n/\eps)$ remains. They used this to prove that there must be no node of out-degree $D+1$ left, at the end of the process, as any such node would imply the existence of an augmenting path of length at most $\ell=O(\log n/\eps)$ \cite[Lemma D.9]{GS17}.

The only algorithmic piece that remains to be explained is how we compute a maximal set of edge-disjoint augmenting paths of length at most $3+i < \ell$, in a given orientation. Ghaffari and Su\cite[Theorem D.4]{GS17} solved this part using Luby's randomized MIS algorithm \cite{luby1986simple}. We instead use our hypergraph maximal matching algorithm. In particular, we view each edge as one vertex of our hypergraph, and each augmenting path of length at most $3+i < \ell$ as one hyperedge of our hypergraph. Then, we invoke \Cref{thm:HypergraphMatching}, which provides us with a maximal set of edge-disjoint augmenting paths. The round complexity of the process is at most $\poly(\ell) \cdot \log^{\log (\log n)/\eps + O(1)} \Delta \cdot \log n$, where the first term $\ell$ is because simulating each hyperedge needs $\ell$ rounds, and the second $\ell$-factor comes from the fact that the degree of the hypergraph may be as large as $\Delta^{\ell}$, which means the related logarithm is at most $\ell\log \Delta$. This is the complexity for each iteration. Since the algorithm has $\ell$ iterations, each time working on an incremented augmenting-path length, the overall complexity is at most  
$\poly(\ell) \cdot \log^{\log ((\log n)/\eps) + O(1)} \Delta \cdot \log n$. This is no more than $2^{O(\log^2 (\log n/\eps)})$ rounds, which is quasi-polylogarithmic in $n$ for most $\eps$-values of interest, e.g., $\eps = \Omega(1/\poly\log n)$.
\end{proof}

\section{Open Problems}
We believe that our techniques and results open the road for further progress on deterministic distributed graph algorithms, with clear consequences also on randomized algorithms, as exemplified by \Cref{thm:Randedge-coloring}. As Barenboim and Elkin suggested when discussing their Open Problem 5, perhaps these will serve as a ``\emph{good stepping stone}" towards obtaining an efficient deterministic algorithms for MIS, thus resolving Linial's long-standing question\cite{linial1987LOCAL}. As concrete steps on this path, we point out two smaller problems, which appear to be the immediate next steps.

\paragraph{Hypergraph Maximal Matching with Better Rank Dependency} For our hypergraph maximal matching algorithm, we have been more focused on the case of smaller ranks $r$, and the current complexity has a factor of $\log^{\log r} \Delta$ in it. Can we improve this to $\poly(r \log n)$, for instance? Notice that the case of $r=\poly\log n$ captures a range of problems of interest, see e.g. \Cref{subsec:orientation}, and this improvement would give a $\poly\log n$-time algorithm for these cases, including a resolution of Open Problem 10 of Barenboim and Elkin's book\cite{barenboim2013monograph}. 

\paragraph{Better than $(2\Delta-1)$-Edge-Coloring} We obtained a polylogarithmic-time algorithm for $(2\Delta-1)$-edge-coloring, as formalized in \Cref{thm:edge-coloring}. This value of $2\Delta-1$ is a natural threshold, because this is what greedy sequential arguments obtain, which made it the classic target of (deterministic) distributed algorithms. However, as Vizing's theorem\cite{vizing1964estimate} shows, every graph has a $(\Delta+1)$-edge-coloring. How close can we get to this, while remaining with polylogarithmic-time \local algorithms?

We are confident that by combining \Cref{thm:edge-coloring} with ideas of Panconesi and Srinivasan~\cite{panconesi1995local} for $\Delta$-vertex-coloring, we can obtain a polylogarithmic-time algorithm for $(2\Delta-2)$-edge-coloring.  But how about $(2\Delta-3)$-edge-coloring, or even $(3\Delta/2)$-edge-coloring?
\medskip

Interestingly, we can already make some progress on this question for graphs with small arboricity. This result is achieved by combining our list-edge-coloring algorithm of \Cref{thm:edge-coloring} with an $H$-partitioning method of Barenboim and Elkin\cite[Chapter 5.1]{barenboim2013monograph}. This significantly generalizes the $(\Delta+o(\Delta))$-edge-coloring results of \cite{barenboim2016edgecoloring}, which worked for $a \leq \Delta^{1-\delta}$ for some constant $\delta>0$.

\begin{corollary}\label{lem:edge-coloring-lowArb} There is a deterministic distributed algorithm that computes an edge-coloring with $\Delta + (2+\eps) a -1$ colors in $O(\frac{1}{\eps}\log^7 \Delta \log^2 n)$ rounds, on any $n$-node graph $G=(V, E)$ with maximum degree $\Delta$ and arboricity $a$. 
\end{corollary}
Notice that any graph has arboricity $a \leq \Delta/2$. The above corollary shows that we start seeing savings in the number of colors as soon as the arboricity goes slightly below this upper bound, e.g., for $a< \Delta (1-\eps)/2$, we already get colorings with less than $2\Delta-2$ colors.

\begin{proof}[Proof of \Cref{lem:edge-coloring-lowArb}]
First, we compute an $H$-partitioning\cite[Chapter 5.1]{barenboim2013monograph} in $O(\log n/\eps)$ rounds. This decomposes $V$ into disjoint vertex sets $H_1$, $H_2$, \dots, $H_\ell$, for $\ell=O(\log n/\eps)$, with the property that each node in $H_i$ has degree at most $(2+\eps) a$ in the graph $G[\cup_{j=i}^{\ell} H_j]$. To compute this decomposition, one just needs to iteratively peel vertices of degree at most $(2+\eps)a$ from the remaining graph.

Having this partitioning, we compute a $(\Delta + (2+\eps) a-1)$-edge-coloring by gradually moving backwards in this partition, from $H_\ell$ towards $H_{1}$. Each step is as follows. Suppose we already have a coloring of edges of $G[\cup_{j=i+1}^{\ell} H_j]$. We now introduce the vertices of $H_i$ and also their edges whose other endpoint is in $\cup_{j=i}^{\ell} H_j$. Each such edge $e$ has at most $(2+\eps)a-1$ other incident edges on the side of its $H_i$-endpoint and at most $\Delta-1$ other incident edges on the other endpoint. If we take away the colors of $\{1, 2, \dots, \Delta + (2+\eps) a-1\}$ that are already used by neighboring edges $e'$ whose both endpoints are in $\cup_{j=i+1}^{\ell} H_j$, the edge $e$ would still have at least $d_{e}+1$ remaining colors in its palette, where $d_e$ is the number of edges in $G[\cup_{j=i}^{\ell} H_j]$ incident on $e$ who remain uncolored. Hence, we can color all these edges by applying the list-edge-coloring algorithm of \Cref{thm:edge-coloring}, in $O(\log^7 \Delta \log n)$ rounds. This is the round complexity needed for coloring new edges after introducing each layer $H_i$. Hence, the overall complexity until we go through all the $\ell$ layers and finish the edge-coloring of $G=G[\cup_{j=1}^{\ell} H_j]$ is $\ell \cdot O(\log^7 \Delta \log n) = O(\frac{1}{\eps}\log^7 \Delta \log^2 n)$.
\end{proof}

\section*{Acknowledgment} We are grateful to Moab Arar and Shiri Chechik for sharing with us their manuscript about distributed matching approximation in graphs\cite{ArarChechik2017Matching}. We note that we arrived at a prior (and slower) version of \Cref{lemma:hypergraphbasicRounding} (for rank-$3$ hypergraphs) inspired by a concept they use, called the \emph{kernel of a graph}, which itself is borrowed from \cite{bhattacharya2016new}.  

\bibliographystyle{alpha}
\bibliography{ref}

\newpage

\begin{appendix}
 
{\noindent\huge\bf Appendix}

\section{Solution of the Recurrence Relation of \Cref{sec:hypergraphmatching,sec:boundedindep}}

\begin{lemma}\label{lem:recursion}
  Let $r\geq 2$ and $\Delta\geq 2$ be two parameters and let $\alpha\geq 1$ and $c>0$ be two
  given constants. Further, let $R(L)$ be a function that is defined
  for $L\geq 1$ by the following recurrence relation:
  \begin{equation}\label{eq:recursion}
    R(L)\ :=\
    \begin{cases}
      cr^2 + c\log\Delta, & \text{ if } L\leq 4,\\
      \alpha r \cdot R(\sqrt{2L}) + c r, & \text{ otherwise.}
    \end{cases}
  \end{equation}
  Then we have
  $R(L) = O\big(r^2 + (\log L)^{\log_2 \alpha + \log_2 r}(r^2 + \log\Delta)\big)$.
\end{lemma}
\begin{proof}
  For all $x\geq 1$, we define a non-negative integer $t_x$ as
  \[
  t_x := \min\set{t \in \mathbb{N}_0 \,:\, \left(\frac{x}{2}\right)^{2^{-t}} \leq 2}.
  \]

  We prove that for all $L\geq 1$, we have
  \begin{equation}\label{eq:recsolution}
    R(L) \leq (\alpha r)^{t_L}\cdot (cr^2 +c\log\Delta) +
    cr\sum_{i=0}^{t_L-1}(\alpha r)^i 
    \ \stackrel{(\alpha r \geq 2)}{<}\  
    2(\alpha r)^{t_L}\cdot (c r^2 + c\log\Delta).
  \end{equation}
  For $x\geq 1$, we have $t_x \leq \max\set{0, \log_2\log_2 x}$ and
  thus the claim of the lemma directly follows from
  \Cref{eq:recsolution}.

  To prove \Cref{eq:recsolution}, first note that for $L\leq 4$, we
  have $t_L\geq 0$ and because $\alpha r \geq 1$, we thus have
  $R(L)\leq c r^2$ as required by \Cref{eq:recursion}. For $L>4$, we
  prove \Cref{eq:recsolution} by induction. More formally, for each
  $L>4$, we show that there is a finite sequence
  $L=L_k > L_{k-1} > \cdots > L_0$ such that $L_0 \leq 4$ and such
  that for each $i\in \set{1,\dots, k}$, \Cref{eq:recursion} implies
  that if \Cref{eq:recsolution} holds for $L_{i-1}$, it also holds
  for $L_{i}$.
  
  Let us therefore assume that $L_k=L>4$. For $i\geq 1$, we define
  $L_{i-1} := \sqrt{2L_{i}}$. First note that because for $x>4$,
  $\sqrt{2x} \leq x/\sqrt{2}$ and thus we reach a value smaller than
  $4$ in a bounded number of steps. For every $i\geq 1$ such that
  $L_{i}>4$, we have
  \[
  t_{L_{i-1}} = \min\set{t\in \mathbb{N}_0 \,:\,
    \left(\sqrt{\frac{L_{i}}{2}}\right)^{2^{-t}} =
    \left(\frac{L_{i}}{2}\right)^{2^{-(t+1)}} \leq 2}\ =\ t_{L_{i}} - 1.
  \]
  From \Cref{eq:recursion}, for $L_{i}>4$, we therefore have
  \begin{eqnarray*}
    R(L_{i}) 
    & \leq & \alpha r \cdot R(L_{i-1}) + cr \\
    & \leq & \alpha r \cdot \left((\alpha r)^{t_{L_{i}}-1}\cdot(cr^2
             + c\log\Delta)
             + cr\cdot
             \sum_{j=0}^{t_{L_i}-2}(\alpha r)^j\right) + cr \\
    & = & (\alpha r)^{t_{L_i}}\cdot (cr^2 +\log\Delta)+
    cr\cdot\sum_{j=0}^{t_{L_i}-1}(\alpha r)^j.
  \end{eqnarray*}
  This proves \Cref{eq:recsolution} and thus concludes the proof.
\end{proof}
\end{appendix}
\end{document}